\documentclass[10pt, conference, letterpaper]{IEEEtran}

\IEEEoverridecommandlockouts
\ifCLASSOPTIONcompsoc
    \usepackage[nocompress]{cite}
\else
    \usepackage{cite}
\fi
\ifCLASSOPTIONcompsoc
    \usepackage[caption=false,font=footnotesize,labelfont=sf,textfont=sf]{subfig}
\else
    \usepackage[caption=false,font=footnotesize]{subfig}
\fi
\ifCLASSOPTIONcaptionsoff
    \usepackage[nomarkers]{endfloat}
\let\MYoriglatexcaption\caption
\renewcommand{\caption}[2][\relax]{\MYoriglatexcaption[#2]{#2}}
\fi
\usepackage[T1]{fontenc}
\usepackage{url}
\usepackage{subfig}
\usepackage{color,xcolor}
\usepackage{dsfont}
\usepackage{bm}
\usepackage[pdftex]{graphicx}
\usepackage{array,amsmath,amssymb,amsfonts}
\usepackage{amsthm}
\usepackage[ruled,lined,linesnumbered,commentsnumbered,longend]{algorithm2e}
\makeatletter
\newcommand{\removelatexerror}{\let\@latex@error\@gobble}
\makeatother

\SetCommentSty{mycommfont}

\renewcommand{\vec}[1]{\boldsymbol{#1}}
\DeclareMathOperator*{\argmin}{argmin}

\DeclareMathOperator*{\diag}{diag}
\newtheorem{theorem}{\textbf{Theorem}}
\newtheorem{lemma}{\textbf{Lemma}}

\def\BibTeX{{\rm B\kern-.05em{\sc i\kern-.025em b}\kern-.08em
    T\kern-.1667em\lower.7ex\hbox{E}\kern-.125emX}}
\vspace{-0.3cm}
\IEEEoverridecommandlockouts

\begin{document}
\title{Placement is not Enough: Embedding with Proactive Stream Mapping on the Heterogenous Edge
}

\author{
    \IEEEauthorblockN{
        Hailiang Zhao\IEEEauthorrefmark{2},
        Shuiguang Deng\IEEEauthorrefmark{2},
        Zijie Liu\IEEEauthorrefmark{2},
        Zhengzhe Xiang\IEEEauthorrefmark{3}, and
        Jianwei Yin\IEEEauthorrefmark{2}
    }
    \IEEEauthorblockA{\IEEEauthorrefmark{2}College of Computer Science and Technology, Zhejiang University}
    \IEEEauthorblockA{\IEEEauthorrefmark{3}College of Computer Science and Technology, Zhejiang University City College}
}

\maketitle

\begin{abstract}
    Edge computing is naturally suited to the applications generated by Internet of Things (IoT) nodes. The IoT applications 
    generally take the form of directed acyclic graphs (DAGs), where vertices represent interdependent functions and edges 
    represent data streams. The status quo of minimizing the makespan of the DAG motivates the study on optimal function 
    placement. However, current approaches lose sight of proactively mapping the data streams to the physical links between 
    the heterogenous edge servers, which could affect the makespan of DAGs significantly. To solve this problem, we study both 
    function placement and stream mapping with data splitting simultaneously, and propose the algorithm DPE (Dynamic Programming-based 
    Embedding). DPE is theoretically verified to achieve the global optimality of the embedding problem. The complexity analysis 
    is also provided. Extensive experiments on Alibaba cluster trace dataset show that DPE significantly outperforms two state-of-the-art 
    joint function placement and task scheduling algorithms in makespan by $43.19$\% and $40.71$\%, respectively.
\end{abstract}


\section{Introduction}\label{s1}
Nowadays, widely used stream processing platforms, such as Apache Spark, Apache Flink, Amazon Kinesis Streams, etc., 
are designed for large-scale data-centers. These platforms are not suitable for real-time and latency-critical applications 
running on widely spread Internet of Things (IoT) nodes. By contrast, near-data processing within the network edge 
is a more applicable way to gain insights, which leads to the birth of edge computing. However, state-of-the-art edgy 
stream processing systems, for example, Amazon Greengrass, Microsoft Azure IoT Edge and so on, do not consider that how the dependent 
functions of the IoT applications distributed to the resource-constrained edge. To address this limitation, works studying 
\textit{function placement} across the distributed edge servers spring up \cite{placement1,placement2,placement3}. In these 
works, the IoT application is structured as a service function chain (SFC) or a directed acyclic graph (DAG) composed of 
interdependent functions, and the placement strategy of each function is obtained by minimizing the makespan of the 
application, under the trade-off between node processing time and cross-node communication overhead. 

However, when minimizing the makespan of the application, state-of-the-art approaches only optimize the placement of functions, 
while passively generate the \textit{stream mapping}. Here the stream mapping refers to mapping the input/output streams to the 
physical links between edge servers. The passivity here means the routing path between each placed function is not optimized but 
generated automatically through SDN controllers. Nevertheless, for the heterogenous edge, where each edge server has different 
processing power and each link has various throughput, better utilization of the stream mapping can result in less makespan even 
though the corresponding function placement is worse. This phenomenon is illustrated in Fig. \ref{fig1}. The top half of this 
figure is an undirected connected graph of six edge servers, abstracted from the physical infrastructure of the heterogenous 
edge. The numbers tagged in each node and beside each link of the undirected graph are the processing power (measured in flop/s) 
and throughput (measured in bit/s), respectively. The bottom half is a SFC with three functions. The number tagged inside 
each function is the required processing power (measured in flops). The number tagged beside each data stream is the size of it 
(measured in bits). Fig. \ref{fig1} demonstrates two solutions of function placement. The numbers tagged beside nodes and links 
of each solution are the time consumed (measured in second). Just in terms of function placement, solution 1 enjoys lower function 
processing time ($2.5$s $<$ $4$s), thus performs better. However, the makespan of solution 2 is $1.5$s lesser than solution 
1 because the path which solution 2 routes through possess a higher throughput. 

\begin{figure}[htbp]
    \centerline{\includegraphics[width=2.2in]{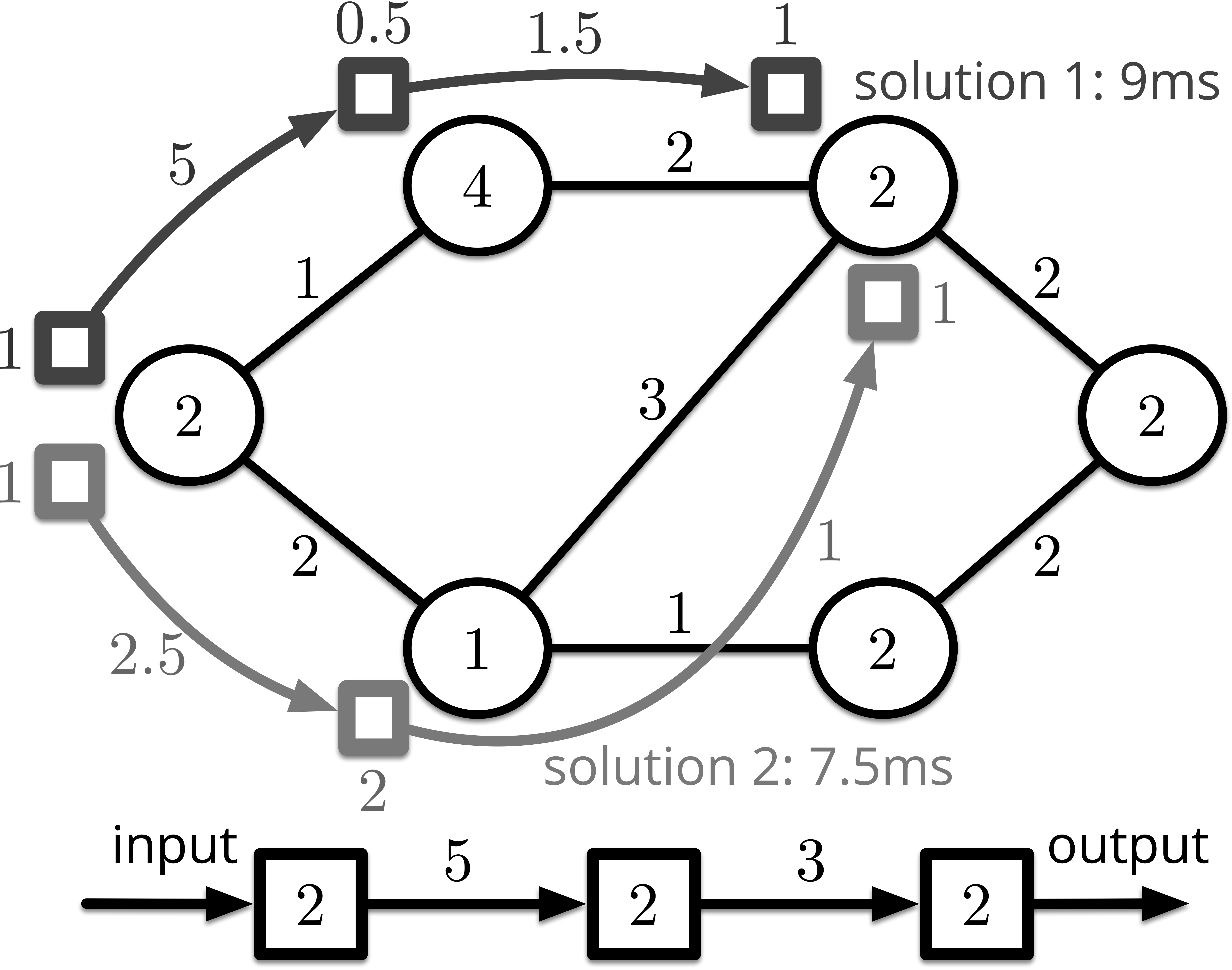}}
    \caption{Two function placement solutions with different stream mappings.}
    \label{fig1}
\end{figure}

The above example implies that different stream mappings can significantly affect the makespan of the IoT applications. It enlightenes 
us to take the stream mapping into consideration \textit{proactively}. In this paper, we name the combination of function placement 
and stream mapping as \textit{function embedding}. Moreover, if stream splitting is allowed, i.e., the output data stream of a function 
can be splitted and route on multiple paths, the makespan can be decreased further. This phenomenon is captured in Fig. \ref{fig2}. The 
structure of this figure is the same as Fig. \ref{fig1}. It demonstrates two function embedding solutions with stream splitting allowed 
or not, respectively. In  solution 2, the output stream of the first function is divided into two parts, with $2$bits 
and $3$bits, respectively. Correspondingly, the time consumed on routing are $3$s and $2.5$s, respectively. Although the two solutions 
have the same function placement, the makespan of solution 2, which is calculated as $1 + \max\{3, 2.5\} + 1 + 1.5 + 1 = 7.5$s, is $4.5$s 
lesser than  solution 1. In practice, segment routing (SR) can be applied to split and route data stream to different edge servers 
by commercial DNS controllers and HTTP proxies \cite{SR}.

\begin{figure}[htbp]
    \centerline{\includegraphics[width=2.2in]{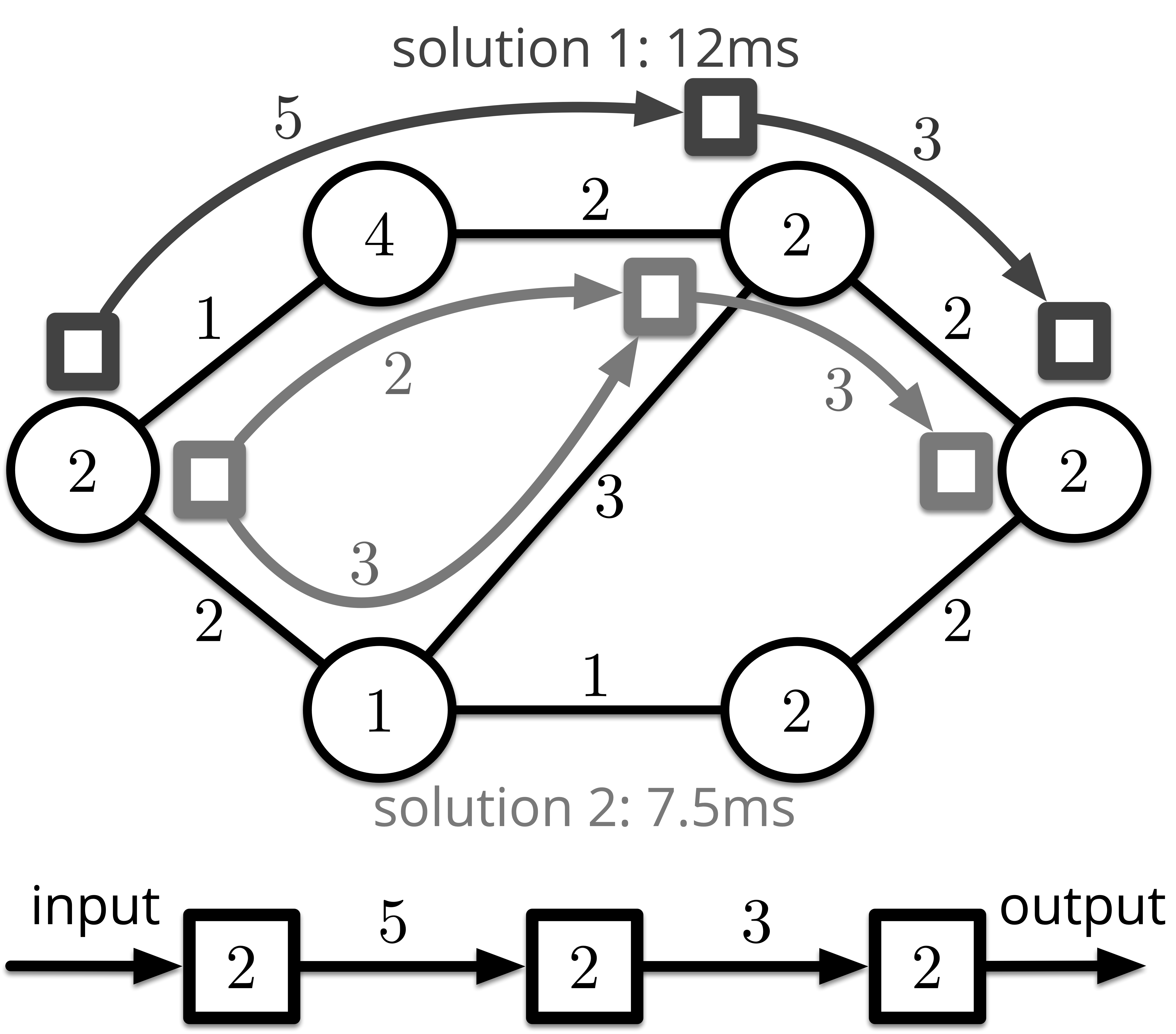}}
    \caption{Two function placement solutions with stream splitting allowed or not, respectively.}
    \label{fig2}
\end{figure}

To capture the importance of stream mapping and embrace the future of 5G communications, in this paper we study the \textit{dependent 
function embedding} problem with \textit{stream splitting} on the heterogenous edge. The problem is similar with 
Virtual Network Embedding (VNE) problem in 5G network slicing for virtualized network functions (VNFs) \cite{slicing-algorithm-analysis}. 
The difference lies in that the VNFs are user planes, management planes and control planes with different levels of granularities, which are virtualized for 
transport and computing resources of telecommunication networks \cite{VNE-survey}, rather than the IoT applications this paper 
refers to. For a DAG with complicated structure, the problem is combinatorial and difficult to solve when it scales up. 
In this paper, we firstly find the optimal substructure of the problem. The basic framework of our algorithm is based on Dynamic Programming. 
For each substructure, we seperate several linear programming sub-problems and solve them optimally. We do not adopt the regular iteration-based 
solvers such as simplex method and dual simplex method, but derive the optimal results directly. 
At length, our paper make the following contributions: 
\begin{itemize}
    \item \textbf{Model contribution:} We study the dependent function embedding problem on the heterogenous edge. Other than existing 
    works where only function placement is studied, we novelly take proactive stream mapping and data splitting into consideration 
    and leverage dynamic programming as the approach to embed DAGs of IoT applications onto the constrained edge.
    \item \textbf{Algorithm contribution:} We present an algorithm that solves the dependent function embedding problem optimally. We firstly 
    find the optimal substructure of the problem. In each substructure, when the placement of each function is fixed, we derive the paths 
    and the data size routes through each path optimally. 
    \item \textbf{Experiment contribution:} We conduct extensive simulations on a cluster trace with 20365 unique DAGs from Alibaba \cite{alibaba}. 
    Experiment results show that our algorithm significantly outperforms two algorithm, FixDoc \cite{placement1} and HEFT \cite{HEFT}, on the average 
    completion time by $43.19$\% and $40.71$\%, respectively.
\end{itemize}

The remainder of the paper is organized as follows. In Sec. \ref{s2}, we present the system model and formulate the problem. In Sec. \ref{s3}, 
we present the proposed algorithms. Performance guarantee and complexity analysis are provided in Sec. \ref{s4}. 
The experiment results are demonstrated in Sec. \ref{s5}. In Sec. \ref{s6}, we review related works on functions placement on the heterogenous edge. 
Sec. \ref{s7} concludes this paper. 

\section{System Model}\label{s2}
Let us formulate the heterogenous edge as an undirected connected graph $\mathcal{G} \triangleq (\mathcal{N}, \mathcal{L})$, where 
$\mathcal{N} \triangleq \{n_1, ..., n_N\}$ is the set of edge servers and $\mathcal{L} \triangleq \{l_1, ..., l_L\}$ is the set of links. 
Each edge server $n \in \mathcal{N}$ has a processing power $\psi_n$, measured in $\textrm{flop/s}$ while each link $ l \in \mathcal{L}$ 
has the same uplink and downlink throughput $b_l$, measured in bit/s. 

\subsection{Application as a DAG}
The IoT application with interdependent functions is modeled as a DAG. The DAG can have arbitrary shape, not just linear SFC. In 
addition, multi-entry multi-exit is allowed. We write $(\mathcal{F}, \mathcal{E})$ for the DAG, where 
$\mathcal{F} \triangleq \{ f_1, ..., f_Q \}$ is the set of $Q$ interdependent functions listed \textit{in topological order}. 
$\forall f_i, f_j \in \mathcal{F}, i \neq j$, if the output stream of $f_i$ is the input of its downstream function $f_j$, a directed link $e_{ij}$ 
exists. $\mathcal{E} \triangleq \{ e_{ij} | \forall f_i, f_j \in \mathcal{F}\}$ is the set of all directed links. For each function 
$f_i \in \mathcal{F}$, we write $c_i$ for the required number of floating point operations of it. For each directed link $e_{ij} \in \mathcal{E}$, 
the data stream size is denoted as $s_{ij}$ (measured in bits). 

\subsection{Dependent Function Embedding}

We write $p(f_i) \in \mathcal{N}$ for the chosen edge server which $f_i$ to be placed on, and $\mathcal{P}(e_{ij})$ for the set of 
paths from $p(f_i)$ to $p(f_j)$. Obviously, for all path $\varrho \in \mathcal{P}(e_{ij})$, it consists of links from $\mathcal{L}$ without 
duplication. For a function pair $(f_i, f_j)$ and its associated directed link $e_{ij} \in \mathcal{E}$, the data 
stream can be splitted and route through different paths from $\mathcal{P}(e_{ij})$. $\forall \varrho \in \mathcal{P}(e_{ij})$, let us 
use $z_{\varrho}$ to represent the allocated data stream size for path $\varrho$. Then, $\forall e_{ij} \in \mathcal{E}$, we have the following constraint:
\begin{equation}
    \sum_{\varrho \in \mathcal{P}(e_{ij})} z_{\varrho} = s_{ij}.
    \label{cons1}
\end{equation}
Notice that if $p(f_i) = p(f_j)$, i.e., $f_i$ and $f_j$ are placed on the same edge server, then $\mathcal{P}(e_{ij}) = \varnothing$ and the routing 
time is zero. Fig. \ref{fig3} gives a working example. The connected graph in it has four edge servers and five links, from $l_1$ to $l_5$. 
The two squares represents the source function $f_i$ and the destination function $f_j$. From the edge server 
$p(f_i)$ to the edge server $p(f_j)$, $e_{ij}$ routes through three paths with data size of $3$ bits, $2$ bits, and $1$ bits, respectively. In this 
example, $s_{ij} = 6$. On closer observation, we can find that two data streams route through $l_1$. Each of them is from path 
$\varrho_1$ and $\varrho_2$ with $3$ bits and $2$ bits, respectively. 
\begin{figure}[htbp]
    \centerline{\includegraphics[width=2.4in]{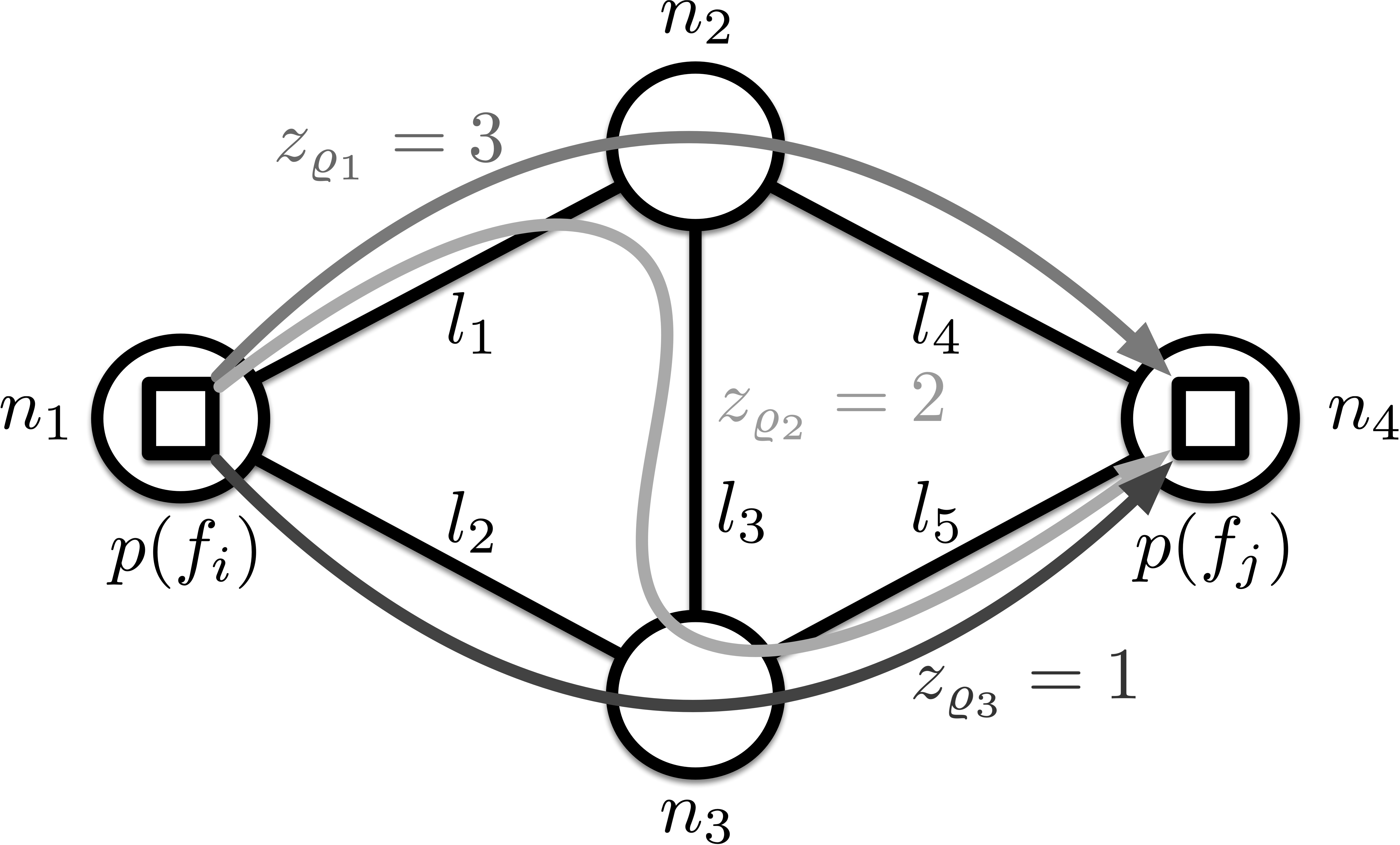}}
    \caption{A working example for stream splitting.}
    \label{fig3}
\end{figure}

\subsection{Involution Function of Finish Time}
Let us use $T \big( p(f_i) \big)$ to denote the finish time of $f_i$ on edge server $p(f_i)$. Considering that the functions of the DAG have 
interdependent relations, for each function pair $(f_i, f_j)$ where $e_{ij}$ is defined, $T \big( p(f_j) \big)$ should involve according to
\begin{equation}
    T \big( p(f_j) \big) = \max_{\forall i: e_{ij} \in \mathcal{E}} \Big( T \big( p(f_i) \big) + t(e_{ij}) \Big) + t \big( p(f_j) \big), 
    \label{involution}
\end{equation}
where $t(e_{ij})$ is the routing time of the directed link $e_{ij}$ and $t \big( p(f_j) \big)$ is the processing time of $f_j$ on edge server $p(f_j)$. 
Corresponding to \eqref{involution}, for each entry function $f_i$, 
\begin{equation}
    T\big( p(f_i) \big) = t \big( p(f_i) \big) + rt_{p(f_i)},
    \label{source}
\end{equation}
where $rt_{p(f_i)}$ stores the beginning time for processing $f_i$. If $f_i$ is the first function which is 
scheduled on server $p(f_i)$, $rt_{p(f_i)}$ is zero. Otherwise, it is the finish time of the last function which is scheduled on server $p(f_i)$. 

$\forall e_{ij} \in \mathcal{E}$, $t(e_{ij})$ is decided by the paths in $\mathcal{P}(e_{ij})$. Specifically, when $f_i$ and $f_j$ are not placed 
on the same edge server, 
\begin{equation}
    t(e_{ij}) = \max_{\varrho \in \mathcal{P}(e_{ij})} \sum_{l \in \varrho} \frac{z_{\varrho}}{b_l},
    \label{link_time}
\end{equation}
where $\mathds{1} \{\cdot\}$ is the indicator function. \eqref{link_time} means that 
the routing time is decided by the slowest branch of data streams. Solution 2 in Fig. \ref{fig2} is an example. 

$\forall f_j \in \mathcal{F}$, $t \big( p(f_j) \big)$ is decided by the processing power of the chosen edge server $p(f_j)$:
\begin{equation}
    t \big( p(f_j) \big) = \frac{c_j}{\psi_{p(f_j)}}.
\end{equation}

To describe the earliest makespan of the DAG, inspired by \cite{placement1}, we add a dummy tail function $f_{Q+1}$. As a dummy function, the processing 
time of $f_{Q+1}$ is set as zero whichever edge server it is placed on. That is, 
\begin{equation}
    t \big( p(f_{Q+1}) \big) = 0, \forall p(f_{Q+1}) \in \mathcal{N}. 
\end{equation}
Besides, each destination function needs to point to $f_{Q+1}$ with a link weight of its output data stream size. 
We write $\mathcal{F}_{dst} \subset \mathcal{F}$ for the set of destination functions and $\mathcal{F}_{src}$ for the set of entry functions of the DAG. 
In addition, we write $\mathcal{E}_{dmy}$ for the set of new added links which point to $f_{Q+1}$, i.e., 
\begin{equation}
    \mathcal{E}_{dmy} \triangleq \{ e_{i,Q+1} | \forall f_i \in \mathcal{F}_{dst} \}.
\end{equation}
As such, the DAG is updated as $(\mathcal{F}', \mathcal{E}')$, where $\mathcal{F}' \triangleq \mathcal{F} \cup \{f_{Q+1}\}$ and 
$\mathcal{E}' \triangleq \mathcal{E} \cup \mathcal{E}_{dmy}$.

\subsection{Problem Formulation}
Our target is to minimize the makespan of the DAG by finding the optimal $p(f_i)$, $\mathcal{P}(e_{ij})$, and $z_{\varrho}$ for all $f_i \in \mathcal{F}'$, 
$e_{ij} \in \mathcal{E}'$, and $\varrho \in \mathcal{P}(e_{ij})$. Let us use $T^\star \big( p(f_i) \big)$ to represent the earliest finish time of $f_i$ 
on edge server $p(f_i)$. Further, we use $ T\big(\mathcal{F}', \mathcal{E}', p(f_{Q+1})\big)$ to represent the earliest makespan of the DAG when $f_{Q+1}$ is 
placed on $p(f_{Q+1})$. Obviously, it is equal to $T^\star \big(p(f_{Q+1}) \big)$. The dependent function embedding problem can be formulated as:
\begin{eqnarray}
    \mathbf{P}: \min_{\forall p(f_i), \forall \mathcal{P}(e_{ij}), \forall z_{\varrho}} T \Big(\mathcal{F}', \mathcal{E}', p(f_{Q+1}) \Big) \nonumber \\
    s.t. \qquad \qquad \qquad \eqref{cons1}, \qquad \qquad \qquad\nonumber \\
    \quad z_{\varrho} \geq 0, \forall \varrho \in \mathcal{P}(e_{ij}), \forall e_{ij} \in \mathcal{E}'. \label{cons2}
\end{eqnarray}

\section{Algorithm Design}\label{s3}
\subsection{Finding Optimal Substructure}
The set of function embedding problems are proved to be NP-hard \cite{slicing-algorithm-analysis}. As a special case of these problems, 
$\mathbf{P}$ is NP-hard, too. because of the dependency relations of fore-and-aft functions, the optimal placement of functions and 
optimal mapping of data streams cannot be obtained \textit{simultaneously}. Nevertheless, we can solve it by finding its optimal substructure. 

Let us dig deeper into the involution equation \eqref{involution}. We can find that 
\begin{equation}
    T \big( p(f_j) \big) = \max_{\forall i: e_{ij} \in \mathcal{E}'} \Big( T \big( p(f_i) \big) + t(e_{ij}) + t \big( p(f_j) \big) \Big)
\end{equation}
because $t \big( p(f_j) \big)$ has no impact on $\max_{\forall i: e_{ij} \in \mathcal{E}'} (\cdot)$. Notice that $\mathcal{E}$ is replaced by $\mathcal{E}'$. 
Then the following expression holds:
\begin{eqnarray}
    \quad T^\star \big( p(f_j) \big) = \max_{\forall i: e_{ij} \in \mathcal{E}'} \Big\{ \min_{p(f_i), \mathcal{P}(e_{ij}), z_{\varrho}} \nonumber \\
    \quad \Big( T^\star \big( p(f_i) \big) + t(e_{ij}) + t \big( p(f_j) \big) \Big) \Big\}.
    \label{optimal_substructure}
\end{eqnarray}
Besides, for all the entry functions $f_i \in \mathcal{F}_{src}$, $T^\star\big( p(f_i) \big)$ is calculated by \eqref{source} without change.

With \eqref{optimal_substructure}, for each function pair $(f_i, f_j)$ where $e_{ij}$ exists, 
we define the sub-problem $\mathbf{P}_{sub}$: 
\begin{eqnarray*}
    \mathbf{P}_{sub}: \min_{p(f_i), \mathcal{P}(e_{ij}), z_{\varrho}} \mathbf{\Phi}_{ij} \triangleq T^\star \big( p(f_i) \big) + t(e_{ij}) + t \big( p(f_j) \big) \\
    s.t. \quad \eqref{cons1}, \eqref{cons2}. \qquad \qquad \qquad \qquad \qquad
\end{eqnarray*}
In $\mathbf{P}_{sub}$, $p(f_j)$ is fixed. We need to decide where $f_i$ shoud be placed and how $e_{ij}$ is mapped. Based on that, 
$T^\star \big( p(f_j) \big)$ can be updated by $\max_{\forall i: e_{ij} \in \mathcal{E}'} \min_{p(f_i), \mathcal{P}(e_{ij}), z_{\varrho}} \mathbf{\Phi}_{ij}$. 
As a result, $\mathbf{P}$ can be solved \textit{optimally} by calculating the earliest finish time of each function in turn. 

The analysis is captured in \textbf{Algorithm 1}, i.e. Dynamic Programming-based Embedding (DPE). 
Firstly, DPE finds all the \textit{simple paths} (paths without loops) between any two edge servers $n_i$ and $n_j$, i.e. $\mathcal{P}(e_{ij})$ 
(Step 2 $\sim$ Step 6). It will be used to calculate the optimal $z_\varrho^\star$ and $\mathcal{P}^\star(e_{ij})$. For all the functions $f_i$ 
who directly point to $f_j$, DPE firstly fixs the placement of $f_j \in \mathcal{N}$. Then, for the function pair $(f_i, f_j)$, $\mathbf{P}_{sub}$ 
is solved and the optimal solution is stored in $\mathbf{\Theta}$ (Step 18 $\sim$ Step 19). If $p^\star(f_i)$ has been decided beforehand, 
$\mathbf{\Phi}^\star_{ij}$ can be directly obtained by finding $\mathcal{P}^\star(e_{ij})$ and $z^\star_{\varrho}$ between $p^\star(f_i)$ and $n$. 
The optimal transmission cost is stored in $t^\star(e_{ij})$ (Step 11 $\sim$ Step 14). The if statement holds for $f_i \in \mathcal{F}'$ when $f_i$ 
is a predecessor of multiple functions. Step 12 is actually a sub-problem of $\mathbf{P}_{sub}$ 
with $p^\star(f_i)$ fixed. Notice that the calculation of the finish time of the entry functions and the other functions are different (Step 16 and 
Step 21, respectively). At last, the global minimal makespan of the DAG is $\min_{p_(f_{Q+1})} T^\star(p(f_{Q+1}))$. The optimal embedding of 
each function can be retrieved from $\mathbf{\Theta}$. 

\begin{figure}[!h]
    \removelatexerror
    \begin{algorithm}[H]
        \caption{DP-based Embedding (DPE)}
        \KwIn{$\mathcal{G}$ and $(\mathcal{F}', \mathcal{E}')$}
        \KwOut{Optimal value and corresponding solution}
        $\mathbf{\Theta} \leftarrow \varnothing$ \\
        \For{each $n_i \in \mathcal{M}$}
        {
            \For{each $n_j \in \mathcal{M} \wedge n_j \neq n_i$}
            {
                Find all the simple paths between $n_i$ and $n_j$ \\
            }
        }
        \For{$j = |\mathcal{F}_{src}| + 1$ to $Q+1$}
        {
            \For{each $n \in \mathcal{N}$}
            {
                $p(f_j) \leftarrow n$ \tcp*[f]{Fix the placement of $f_j$}\\
                \For{each $f_i \in \{ f_i | e_{ij}\textrm{ exists} \}$}
                {
                    \If{$p^\star(f_i)$ has been decided}
                    {
                        $\mathbf{\Phi}^\star_{ij} \leftarrow T^\star \big( p^\star(f_i) \big) + t^\star(e_{ij}) + t \big( p(f_j) \big)$ \\
                        \textbf{goto} Step 19 \\
                    }
                    \If{$f_i \in \mathcal{F}_{src}$}
                    {
                        $\forall p(f_i) \in \mathcal{N}$, update $T^\star \big( p (f_i) \big)$ by \eqref{source}\\
                    }
                    Obtain the optimal $\mathbf{\Phi}^\star_{ij}$, $p^\star(f_i)$, $\mathcal{P}^\star(e_{ij})$, and $z^\star_\varrho$ by 
                    solving $\mathbf{P}_{sub}$ \\
                    $\vec{\Theta}.append \Big(
                        \big\{p^\star(f_i)\big\} \times 
                        \big\{\mathcal{P}^\star(e_{ij})\big\} \times 
                        \big\{z^\star_{\varrho} | \forall \varrho \in \mathcal{P}^\star(e_{ij})\big\}
                        \Big)$ \\
                }
                Update $T^\star \big( p(f_j) \big)$ by \eqref{optimal_substructure} \\
            }
        }
        \Return{$\min_{p_(f_{Q+1})} T^\star(p(f_{Q+1}))$ and $\mathbf{\Theta}$}
    \end{algorithm}
\end{figure}

\subsection{Optimal Proactive Stream Mapping}
Now the problem lies in that how to find all the simple paths and solve $\mathbf{P}_{sub}$ optimally. 
For the former, we propose a Recursion-based Path Finding (RPF) algorithm, which will be detailed in Sec. \ref{3.B.1}. For the latter, 
When $p(f_j)$ is fixed, the value of $t \big( p(f_j) \big)$ is known 
and can be viewed as a constant for $\mathbf{P}_{sub}$. Besides, from Step 12 of DPE we can find that, for all 
$f_i \in \{f_i \in \mathcal{F}'-\mathcal{F}_{src} | e_{ij}\textrm{ exists}\}$, $\forall p(f_i) \in \mathcal{N}$, $T^\star \big( p (f_i) \big)$ 
is already updated in the last iteration. Thus, for solving $\mathbf{P}_{sub}$, the difficulty lies in that how to select the optimal placement 
of $f_i$ and the optimal mapping of $e_{ij}$. It will be detailed in Sec. \ref{3.B.2}. 

\subsubsection{\textbf{Recursion-based Path Finding}} \label{3.B.1}
We use $\mathcal{P}(n_i, n_j, \mathcal{M})$ to represent the set of simple paths from $n_i$ to $n_j$ where no path goes through nodes from the set 
$\mathcal{M} \subseteq \mathcal{N}$. The set of simple paths from $n_i$ to $n_j$ we want, i.e. $\mathcal{P} (e_{ij})$, is equal to 
$\mathcal{P}(n_i, n_j, \varnothing)$. $\forall n \in \mathcal{N}$, let us use $\mathcal{A}(n)$ to represent the set of edge servers adjcent to $n$. 
Then, $\mathcal{P}(n_i, n_j, \mathcal{M})$ can be calculated by the following recursion formula: 
\begin{equation*}
    \mathcal{P} (n_i, n_j, \mathcal{M}) = 
    \Big\{ J(\varrho, n_i) \Big| \bigcup_{m \in \mathcal{S}} \mathcal{P} \big(m, n_j, \mathcal{M} \cup \{n_i\} \big) \Big\}, 
\end{equation*}
where $\mathcal{S} \triangleq \mathcal{A}(n_i) - \mathcal{M} \cup \{n_i\}$. $J(\varrho, n_i)$ is a function that joins the node $n_i$ 
to the path $\varrho$ and returns the new joint path $\varrho \cup \{n_i\}$. The analysis in this paragraph is summarized in \textbf{Algorithm 2}, 
i.e. Recursion-based Path Finding (RPF) algorithm. 

Before calling RPF, we need to initialize the global variables. Specifically, $\mathbf{\Omega}$ stores all the simple paths, which 
is initialized as $\varnothing$. $\mathcal{V}$, as the set of visited nodes, is initialized as $\varnothing$. $\varrho$ is allocated for 
the storage of current path, which is also initialized as $\varnothing$. Whereafter, by calling $\texttt{RPF}(n_i, n_i, n_j)$, all the simple 
paths between $n_i$ and $n_j$ are stored in $\mathbf{\Omega}$. $\texttt{RPF}(n_i , n_i, n_j)$ is used to replace Step 4 of DPE. 

\begin{figure}[!h]
    \removelatexerror
    \begin{algorithm}[H]
        \caption{Recursion-based Path Finding (RPF)}
        \KwIn{$n_i$, $n$, and $n_j \in \mathcal{N}$}
        \tcc{Global variables $n_i, n_j, \mathcal{V}, \varrho, \mathcal{G},$ and $\mathbf{\Omega}$ can be visited by RPF. 
            Before calling it, $\mathcal{V}$, $\varrho$, and $\mathbf{\Omega}$ are set as $\varnothing$.}
        \uIf{$n == n_j$}
        {
            $\mathbf{\Omega}.append \big( J(\varrho, n) \big)$ \tcp*[f]{Store the path $J(\varrho, n)$}\\
        }
        \Else
        {
            $\varrho.push(n)$; $\mathcal{V}.add(n)$ \\
            \For{each $n' \in \mathcal{A}(n) -\mathcal{V}$}
            {
                $\texttt{RPF}(n_i, n', n_j)$ \tcp*[f]{Recursive call} \\
            }
            $\varrho.pop()$; $\mathcal{V}.delete(n)$ \\
        }
    \end{algorithm}
\end{figure}

In the following, we calculate the optimal value of $z_\varrho$ for each simple path $\varrho \in \mathbf{\Omega}$. 

\subsubsection{\textbf{Optimal Data Splitting}} \label{3.B.2}
Since both $p(f_i)$ and $p(f_j)$ are fixed, $T^\star \big( p(f_i) \big)$ and $t \big( p(f_i) \big)$ are constants. Therefore, solving $\mathbf{P}_{sub}$ 
is equal to solving the following problem:
\begin{eqnarray}
    \mathbf{P}'_{sub}: & \min_{\forall \varrho \in \mathbf{\Omega}: z_\varrho} 
    \max_{\varrho \in \mathbf{\Omega}} \bigg\{ \Big( \sum_{l \in \varrho} \frac{1}{b_l} \Big) \cdot z_\varrho \bigg\} \nonumber \\
    & s.t. \quad
    \left\{ 
    \begin{array}{c}
        \sum_{\varrho \in \mathbf{\Omega}} z_\varrho = s_{ij},\\
        z_\varrho \geq 0, \forall \varrho \in \mathbf{\Omega}.
    \end{array}
    \right. \qquad \qquad \label{P_sub_cons}
\end{eqnarray}
\eqref{P_sub_cons} is reconstructed from \eqref{cons1} and \eqref{cons2}. To solve $\mathbf{P}'_{sub}$, we define a diagonal matrix 
\begin{equation*}
    \mathbf{A} \triangleq \diag \bigg(
        \sum_{l_1 \in \varrho_1} \frac{1}{b_{l_1}}, 
        \sum_{l_2 \in \varrho_2} \frac{1}{b_{l_2}}, ..., 
        \sum_{l_{|\mathbf{\Omega}|} \in \varrho_{|\mathbf{\Omega}|}} \frac{1}{b_{l_{|\mathbf{\Omega}|}}}
        \bigg).
\end{equation*}
Obviously, all the diagonal elements of $\mathbf{A}$ are positive real numbers. The variables that need to be determined can be written 
as $\mathbf{z} \triangleq [z_{\varrho_{1}}, z_{\varrho_2}, ..., z_{\varrho_{|\mathbf{\Omega}|}}]^\top \in \mathbb{R}^{|\mathbf{\Omega}|}$. Thus, 
$\mathbf{P}'_{sub}$ can be transformed into 
\begin{eqnarray}
    \mathbf{P}_{norm}: & \min_{\mathbf{z}} \Vert \mathbf{A} \mathbf{z} \Vert_{\infty} \qquad \quad \nonumber \\
    s.t. &
    \left\{
    \begin{array}{c}
        \mathbf{1}^\top \mathbf{z} = s_{ij}, \\
        \mathbf{z} \geq \mathbf{0}.
    \end{array}
    \right. \qquad \quad \label{P_map_cons}
\end{eqnarray}
$\mathbf{P}_{norm}$ is an infinity norm minimization problem. By introducing slack variables $\tau \in \mathbb{R}$ and 
$\mathbf{y} \in \mathbb{R}^{|\mathbf{\Omega}|}$, $\mathbf{P}_{norm}$ can be transformed into the following linear programming problem: 
\begin{eqnarray*}
    \mathbf{P}'_{norm}: \min_{\mathbf{z}' \triangleq [\mathbf{z}^\top, \mathbf{y}^\top]^\top} \tau \qquad \\
    s.t. \quad
    \left\{
    \begin{array}{c}
        \sum_{\varrho \in \mathbf{\Omega}} z_\varrho = s_{ij},\\
        \mathbf{A} \mathbf{z} + \mathbf{y} = \tau \cdot \mathbf{1}, \\
        \mathbf{z}' \geq \mathbf{0}.
    \end{array}
    \right. 
\end{eqnarray*}
$\mathbf{P}'_{norm}$ is feasible and its optimal objective value is finite. As a result, simplex method and dual simplex method can be 
applied to obtain the optimal solution efficiently. 

However, simplex methods might be time-consuming when the scale of $\mathcal{G}$ increases. In fact, we can find that the optimal 
objective value of $\mathbf{P}_{norm}$ is 
\begin{equation}
    \min_{\mathbf{z}} \Vert \mathbf{A} \mathbf{z} \Vert_{\infty} = \frac{s_{ij}}{\sum_{k=1}^{|\mathbf{\Omega}|} 1/A_{k,k}},
    \label{optimal_obj}
\end{equation}
if and only if 
\begin{equation}
    A_{u,u} \mathbf{z}^{(u)} = A_{v,v} \mathbf{z}^{(v)}, 1 \leq u \neq v \leq |\mathbf{\Omega}|, 
    \label{optimal_var}
\end{equation}
where $\mathbf{z}^{(u)}$ is the $u$-th component of vector $\mathbf{z}$ and $A_{u,u}$ is the $u$-th diagonal element of $\mathbf{A}$. 
Detailed proof of this result is provided in Sec. \ref{s3.C}. Base on \eqref{optimal_var}, we can infer that the optimal variable 
$\mathbf{z}^\star > \mathbf{0}$, which means that $\forall \varrho \in \mathbf{\Omega}$, $z^\star_\varrho \neq 0$. Therefore, the assumption 
in Sec. \ref{3.B.1} is not violated and the optimal $\mathcal{P}^\star(e_{ij})$ is equivalent to $\mathbf{\Omega}$.

\begin{figure}[!h]
    \removelatexerror
    \begin{algorithm}[H]
        \caption{Optimal Stream Mapping (OSM)}
        \KwIn{$\mathcal{G}$, $(\mathcal{F}', \mathcal{E}')$, and $p(f_j)$}
        \KwOut{The optimal $\mathbf{\Phi}^\star_{ij}$, $p^\star(f_i)$, $\mathcal{P}^\star(e_{ij})$, and $\mathbf{z}^\star$}
        \For(\textbf{in parallel}){each $m \in \mathcal{N}$}
        {
            $p(f_i) \leftarrow m$ \\
            \tcc{Obtain the $m$-th optimal $\mathbf{\Phi}_{ij}$ by \eqref{optimal_obj}}
            $\mathbf{\Phi}_{ij}^{(m)} \leftarrow \frac{s_{ij}}{\sum_{k} 1/A_{k,k}^{(m)}} + T^\star \big( p(f_i) \big) + t \big( p(f_j) \big)$ \\
        }
        $p^\star(f_i) \leftarrow \argmin_{m \in \mathcal{N}} \vec{\Phi}_{ij}^{(m)}$ \\
        $\mathcal{P}^\star(e_{ij}) \leftarrow \mathbf{\Omega}^{(p^\star(f_i))}$ \\
        Calculate $\mathbf{z}^\star$ by \eqref{P_map_cons} and \eqref{optimal_var} with $\mathbf{A} = \mathbf{A}^{(p^\star(f_i))}$ \\
        \Return{$\mathbf{\Phi}^{(p^\star(f_i))}_{ij}$, $p^\star(f_i)$, $\mathcal{P}^\star(e_{ij})$, and $\mathbf{z}^\star$}
    \end{algorithm}
\end{figure}

Up to now, when $p(f_i)$ is fixed, we have calculate the optimal $\mathcal{P}^\star(e_{ij})$ and $z^\star_\varrho$ for all paths 
$\varrho \in \mathcal{P}^\star (e_{ij})$. 
The analysis in this subsection are summarized in \textbf{Algorithm 3}, Optimal Stream Mapping (OSM) algorithm. In OSM, $\mathbf{\Phi}_{ij}^{(m)}$ is 
the $m$-th objective value of $\mathbf{P}_{sub}$ by taking $p(f_i) = m$. Similarly, $\mathbf{\Omega}^{(p^\star (f_i))}$ is the $p^\star (f_i)$-th set of 
simple paths obtained by RPF. For at most $|\mathcal{N}_i|$ choices of $p(f_i)$, OSM calculates the optimal objective value (Step 1 $\sim$ Step 5). 
The procedure is executed in parallel (with different threads) because intercoupling is nonexistent. Then, OSM finds the best placement of $f_i$ and 
returns the corresponding $\mathcal{P}^\star(e_{ij})$, $\mathbf{z}^\star$. OSM is used to replace Step 18 of DPE. 

\begin{figure}[htbp]
    \centerline{\includegraphics[width=3.4in]{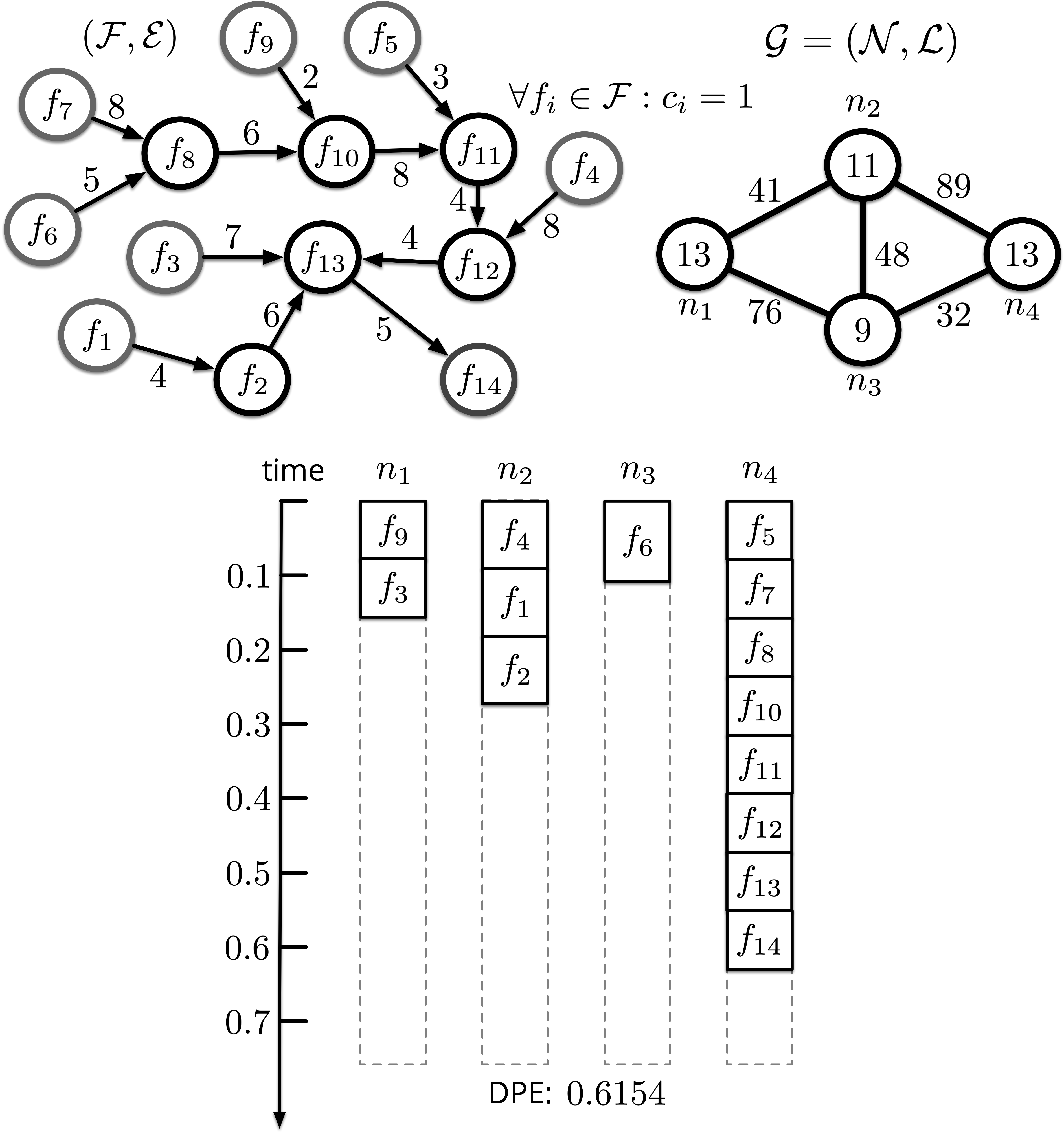}}
    \caption{Embedding of a DAG with the DPE algorithm.}
    \label{fig_add1}
\end{figure}

Fig. \ref{fig_add1} demonstrates an example on how PDE works. The top left portion of the figure is a DAG randomly sampled from the Alibaba 
cluster trace, where all the functions are named in the manner of topological order. $\forall f_i \in \mathcal{F}$, $c_i$ is set as $1$. 
The top right is the edge server cluster $\mathcal{G}$. The bottom demonstrates how the functions are placed and scheduled by DPE. 

\section{Theoretical Analysis}\label{s4}
\subsection{Performance Guarantee}\label{s3.C}
In this subsection, we analyze the optimality of the proposed algorithm, DPE.  
\begin{theorem}
    \textit{\textbf{Optimality of DPE} For a DAG $(\mathcal{F}, \mathcal{E})$ with a given topological order, DPE can achieve the global optimality 
    of $\mathbf{P}$ by replacing the Step 4 of it with RPF and Step 18 of it with OSM.}
\end{theorem}
\begin{proof}
    We firstly prove that OSM solves $\mathbf{P}_{sub}$ optimally. Recall that solving $\mathcal{P}_{sub}$ is equal to solving $\mathbf{P}'_{sub}$, and 
    $\mathbf{P}'_{sub}$ can be transformed into $\mathbf{P}_{norm}$. For $\mathbf{P}_{norm}$, we have
    \begin{equation*}
        \Vert \mathbf{A} \mathbf{z} \Vert_\infty \triangleq \max_k \big\{ |A_{k,k} \mathbf{z}^{(k)}| \big\} = 
        \lim_{x \to \infty} \sqrt[x]{\sum_{k=1}^{|\mathbf{\Omega}|} (A_{k,k} \mathbf{z}^{(k)} )^x}
    \end{equation*}
    because $\forall k, A_{k,k} > 0, \mathbf{z}^{(k)} \geq 0$. 
    According to the \textit{AM–GM inequality}, the following inequality always holds: 
    \begin{equation}
        \frac{\sum_{k=1}^{|\mathbf{\Omega}|} (A_{k,k} \mathbf{z}^{(k)} )^x}{|\mathbf{\Omega}|} \geq 
        \sqrt[|\mathbf{\Omega}|]{\prod_{k=1}^{|\mathbf{\Omega}|}(A_{k,k} \mathbf{z}^{(k)} )^x}, 
        \label{ineq1}
    \end{equation}
    \textit{iff} \eqref{optimal_var} is satisfied. It yields that $\forall x > 0$, 
    \begin{eqnarray}
        \sqrt[x]{\frac{\sum_{k=1}^{|\mathbf{\Omega}|} (A_{k,k} \mathbf{z}^{(k)} )^x}{|\mathbf{\Omega}|}} \geq 
        \sqrt[|\mathbf{\Omega}|]{\prod_{k=1}^{|\mathbf{\Omega}|} A_{k,k} \mathbf{z}^{(k)}}.
        \label{ineq2}
    \end{eqnarray}
    Multiply both sides of \eqref{ineq2} by $\sqrt[x]{|\mathbf{\Omega}|}$, we have 
    \begin{equation}
        \sqrt[x]{\sum_{k=1}^{|\mathbf{\Omega}|} (A_{k,k} \mathbf{z}^{(k)} )^x} \geq 
        \sqrt[x]{|\mathbf{\Omega}|} \cdot \sqrt[|\mathbf{\Omega}|]{\prod_{k=1}^{|\mathbf{\Omega}|} A_{k,k} \mathbf{z}^{(k)}}.
        \label{ineq3}
    \end{equation}
    By taking the limit of \eqref{ineq3}, we have 
    \begin{equation}
        \Vert \mathbf{A} \mathbf{z} \Vert_\infty \geq
        \lim_{x \to \infty} \sqrt[x]{|\mathbf{\Omega}|} \cdot \sqrt[|\mathbf{\Omega}|]{\prod_{k=1}^{|\mathbf{\Omega}|} A_{k,k} \mathbf{z}^{(k)}}.
        \label{ineq4}
    \end{equation}
    Combining with \eqref{P_map_cons} and \eqref{optimal_var}, the right side of \eqref{ineq4} is actually a constant. In other words, 
    \begin{eqnarray}
        \min_{\mathbf{z}} \Vert \mathbf{A} \mathbf{z} \Vert_\infty &=& 
        \lim_{x \to \infty} \sqrt[x]{|\mathbf{\Omega}|} \cdot \sqrt[|\mathbf{\Omega}|]{\prod_{k=1}^{|\mathbf{\Omega}|} A_{k,k} \mathbf{z}^{(k)}} \nonumber \\
        &=& \sqrt[|\mathbf{\Omega}|]{\prod_{k=1}^{|\mathbf{\Omega}|} A_{k,k} \mathbf{z}^{(k)}} \quad \rhd \textrm{with \eqref{P_map_cons}} \nonumber \\
        &=& \frac{s_{ij}}{\sum_{k=1}^{|\mathbf{\Omega}|} 1/A_{k,k}}. \nonumber
    \end{eqnarray}
    The result shows that \eqref{optimal_obj} and \eqref{optimal_var} are the optimal objective value and corresponding optimal 
    condition of $\mathbf{P}_{sub}$, respectively, if the topological ordering is given and regarded as an known variable. 
    The theorem is immediate from the optimality of DP. 
\end{proof}

\subsection{Complexity Analysis}
In this subsection, we analyze the complexity of the proposed algorithms in the worst case, where $\mathcal{G}$ is fully connected.

\subsubsection{\textbf{Complexity of RPF}} 
Let us use $\kappa(i, j)$ to denote the flops required to compute all the simple paths between $n_i$ and $n_j$. 
If $\mathcal{G}$ is fully connected, 
\begin{eqnarray*}
    \kappa(1, N) &=& \big(1 + \kappa(2, N)\big) + \big(1 + \kappa(3, N)\big)\\
    &+& ... + \big(1 + \kappa(N-1, N)\big) + 1\\
    &=& (N - 1) + (N - 2) \cdot \kappa(2, N).
\end{eqnarray*}
To simplify notations, we use $\kappa_i$ to replace $\kappa(i, N)$. We can conclude that $\forall i \in \{1, ..., N-1\}$, 
\begin{eqnarray}
    \kappa_i = (N - i) + (N - i - 1) \cdot \kappa_{i+1}.
    \label{series}
\end{eqnarray}
Based on \eqref{series}, we have 
\begin{eqnarray}
    \kappa_1 &=& (N-2)! \cdot \sum_{i=1}^{N-1} \frac{N-i}{(N-i-1)!} \nonumber \\
    &=& (N-2)! \cdot \sum_{i=1}^{N-1} \frac{i^2}{i!}
    \label{formula}
\end{eqnarray}
which is the maximum flops required to compute all the simple paths between any two edge servers. 
Before calculating the complexity of RPF, we prove two lemmas.

\begin{lemma}
    \textit{$\forall N \geq 7$ and $N \in \mathbb{N}^+$, $N! > N^3 (N+1)$.}
    \label{lemma1}
\end{lemma}
\begin{proof}
    The proof is based on induction. When $N=7$, $N! = 5040 > N^3 (N+1) = 2744$. The lemma holds. Assume that the lemma holds for $N=q$, i.e., 
    $q!> q^3 (q+1)$ (induction hypothesis). Then, for $N=q+1$, we have 
    \begin{equation}
        (q+1)! = (q+1) \cdot q! > (q+1)^2 q^3.
        \label{lemma1_1}
    \end{equation}
    Notice that the function $g(q) \triangleq (\frac{1}{q} + \frac{2}{q^2} + \frac{4}{q^3})^{-1}$ monotonically increases when $q \in \mathbb{N}^+ - \{1, 2\}$. 
    Hence $g(q) \geq g(3) = \frac{27}{25} > 1$, and 
    \begin{eqnarray}
        1 &<& \frac{q^3}{(q+2)^2} < \frac{q^3}{(q+1)(q+2)} \nonumber \\
        &\Rightarrow& q^3 > (q+1)(q+2) \nonumber \\
        &\Rightarrow& (q+1) \cdot q^3 \cdot (q+1) > (q+1)^3 (q+2).
        \label{lemma1_2}
    \end{eqnarray}
    Combining \eqref{lemma1_1} with \eqref{lemma1_2}, we have
    \begin{equation*}
        (q+1)! > (q+1)^3 (q+2), 
    \end{equation*}
    which means the lemma holds for $q+1$. 
\end{proof}

\begin{lemma}
    \textit{$\forall N \geq 2$ and $N \in \mathbb{N}^+$, $\sum_{i=1}^{N-1} \frac{i^2}{i!} < 6 - \frac{1}{N}$.}
    \label{lemma2}
\end{lemma}
\begin{proof}
    We can verify that when $N \in [2, 7] \cap \mathbb{N}^+$, the lamma holds. In the following we prove the lamma holds for $N > 7$ by induction. 
    Assume that the lemma holds for $N = q$, i.e., $\sum_{i=1}^{q-1} \frac{i^2}{i!} < 6 - \frac{1}{q}$ (induction hypothesis). Then, for $N=q+1$, we have 
    \begin{eqnarray}
        \sum_{i=1}^{q} \frac{i^2}{i!} < 6 - \frac{1}{q} + \frac{q^2}{q!}.
    \end{eqnarray}
    By applying Lemma \ref{lemma1}, we get 
    \begin{equation*}
        \sum_{i=1}^{q} \frac{i^2}{i!} < 6 - \frac{1}{q+1}, 
    \end{equation*}
    which means the lemma holds for $q+1$. 
\end{proof}
Based on these lemmas, we can obtain the complexity of RPF, as illustrated in the following theorem:
\begin{theorem}
    \textit{\textbf{Complexity of RPF} In worst case, where $\mathcal{G}$ is a fully connected graph and $N \geq 2$, the complexity of OSM 
    is $O \big( (N-2)! \big)$.}
    \label{theorem2}
\end{theorem}
\begin{proof}
    According to Lemma \ref{lemma2}, 
    \begin{equation*}
        \lim_{N \to \infty} \sum_{i=1}^{N-1} \frac{i^2}{i!} < 6.
    \end{equation*}
    Hence $\lim_{N \to \infty} \kappa_1 < 6 (N-2)! = O \big( (N-2)! \big)$. 
\end{proof}
Finding all the simple paths between arbitrary two nodes is a NP-hard problem. To solve it, RPF is based on depth-first search. 
In real-world edge computing scenario for IoT stream processing, $\mathcal{G}$ might not be fully connected. Even though, the number of edge servers 
is small. Thus, the real complexity is much lower. 

\subsubsection{\textbf{Complexity of DPE}}
Notice that RPF is called by OSM $N$ times in parallel. It is easy to verify that the complexity of OSM is $O \big( (N-2)! \big)$ in worst case, too. 
OSM is designed to replaced the Step 6 of DPE. Thus, we have the following theorem: 
\begin{theorem}
    \textit{\textbf{Complexity of DPE} In worst case, where $\mathcal{G}$ is a fully connected graph and $N \geq 2$, the complexity of DPE 
    is 
    \begin{equation*}
        \max \bigg\{ O( N! ), O\Big( |\mathcal{E}| \cdot N \cdot |\mathcal{P}^\star (e_{ij})| \Big)\bigg\}.
    \end{equation*}}
\end{theorem}
\begin{proof}
    From Step 2 to Step 6, DPE calls RPF $N(N-1)$ times. Thus, the complexity of this part (Step 2 $\sim$ Step 6) is $O(N!)$ according to 
    Theorem \ref{theorem2}. The average number of pre-order functions for each \textit{non-source} function 
    $f \in \mathcal{F}' - \mathcal{F}_{src}$ is $\frac{|\mathcal{E}|}{Q-|\mathcal{F}_{src}|}$. As a result, in average, OSM is called 
    $N \times |\mathcal{E}|$ times. In OSM, the step required the most flops is Step 8. If the variable substitution 
    method is adopted, the flops required of this step is $2(|\mathcal{P}^\star(e_{ij})|-1)+1$. 
    Thus, the complexity of Step 7 $\sim$ Step 19 of DPE is $O\Big( |\mathcal{E}| \cdot N \cdot |\mathcal{P}^\star (e_{ij})| \Big)$. 
    The theorem is immediate by combining the two parts. 
\end{proof}
Although $O(N!)$ is of great order of complexity, $N$ is not too large in real-world edgy scenario. Even if it is not, as an offline algorithm, 
it is worth the sacrifice of runtime overhead in pursuit of global optimality.

\section{Experimental Validation}\label{s5}
\subsection{Experiment Setup}
\textbf{IoT stream processing workloads.} The simulation is conducted based on Alibaba's cluster trace of data analysis. This dataset 
contains more than 3 million jobs (called applications in this work), and 20365 jobs with unique DAG information. Considering that there are 
too many DAGs with only single-digit functions, we sampled 2119 DAGs with different size from the dataset. The distribution of the samples 
is visualized in Fig. \ref{fig_exp1}. For each $f \in \mathcal{F}$, the processing power required and output data size are extracted 
from the corresponding job in the dataset and scaled to $[1, 10]$ giga flop and $[5, 15] \times 10^3$ kbits, respectively. 

\begin{figure}[htbp]
    \centerline{\includegraphics[width=2.5in]{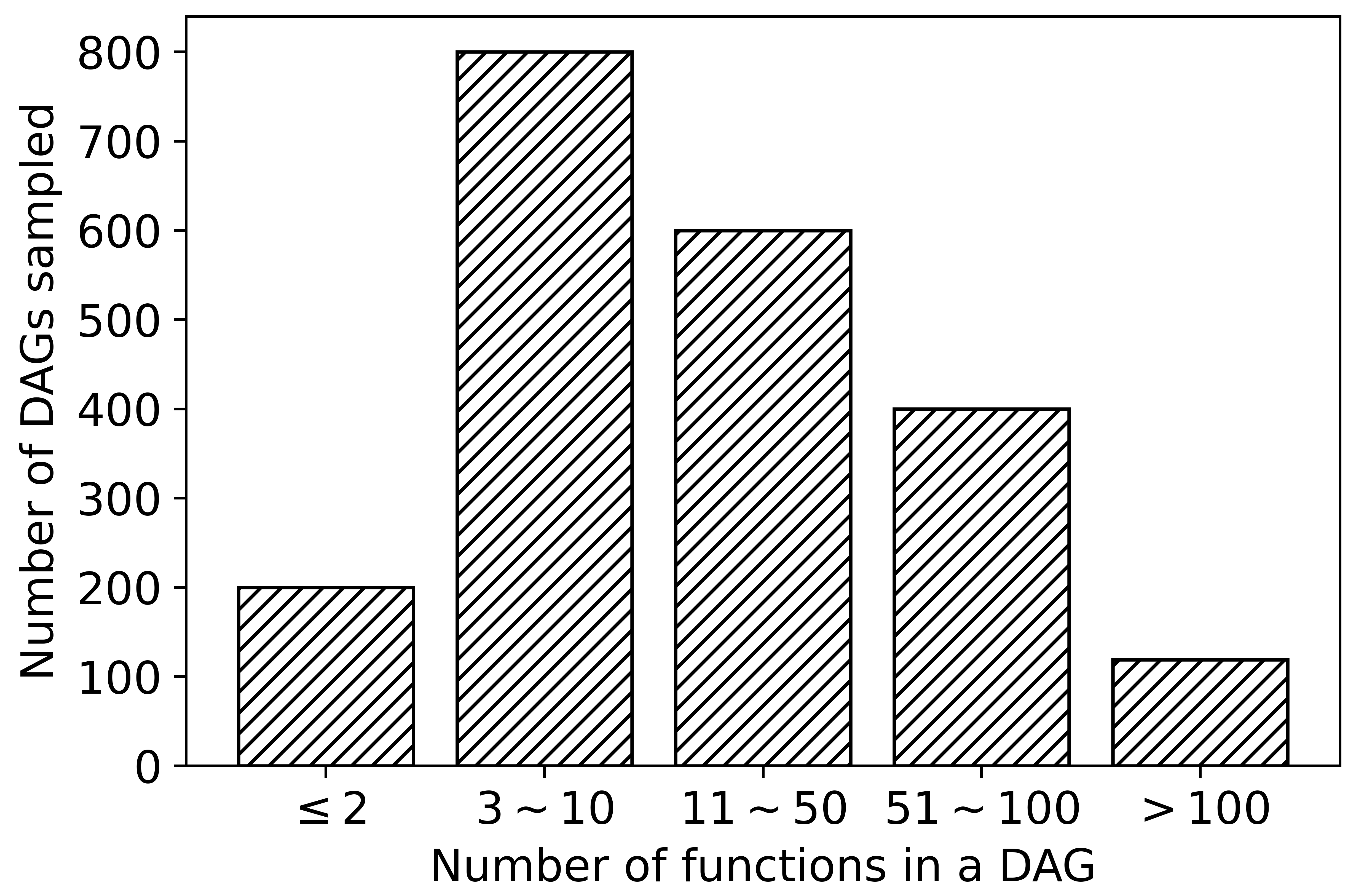}}
    \caption{Data distribution sampled from the cluster trace.}
    \label{fig_exp1}
\end{figure}

\textbf{Heterogenous edge servers.} In our simulation, the number of edge servers is $10$ in default. 
Considering that the edge servers are required to formulate a \textit{connected} graph, the impact of the sparsity of the 
graph is also studied. The processing power of edge servers and the maximum throughput of physical links are uniformly sampled from 
$[20, 40]$ giga flop and $[30, 80] \times 10^3$ kbit/s in default, respectively. 

\textbf{Algorithms compared.} We compare DPE with the following algorithms.
\begin{itemize}
    \item \textit{FixDoc} \cite{placement1}: FixDoc is a function placement and DAG scheduling algorithm with fixed function 
    configuration. FixDoc places each function onto \textit{homogeneous} edge servers optimally to minimize the DAG completion time. 
    Actually, \cite{placement1} also proposes an improved version, GenDoc, with function configuration optimized, too. However, for 
    IoT streaming processing scenario, on-demand function configuration is not applicative. Thus, we only compare DPE with FixDoc. 
    \item \textit{Heterogeneous Earliest-Finish-Time (HEFT)} \cite{HEFT}: HEFT is a heuristic to schedule a set of dependent functions  
    onto heterogenous workers with communication time taken into account. Starting with the highest priority, functions are assigned 
    to different workers to heuristically minimize the overall completion time. HEFT is an algorithm that stands the test of time.
\end{itemize}

\subsection{Experiment Results}
All the experiments are implemented in Python 3.7 on macOS Catalina equipped with 3.1 GHz Quad-Core Intel Core i7 and 16 GB RAM.
\subsubsection{Theoretical Performance Verification}
Fig. \ref{fig_exp2} illustrates the overall performance of the three algorithms. For different data batch, DPE can reduce 
$43.19\%$ and $40.71\%$ of the completion time on average over FixDoc and HEFT on 2119 DAGs. The advantage of DPE is more obvious when the scale of DAG is large 
because the parallelism is fully guaranteed. Fig. \ref{fig_exp3} shows the accumulative distribution of 2119 DAGs' completion time. 
DPE is superior to HEFT and FixDoc on $100\%$ of the DAGs. Specifically, the maximum completion of DAG achieved by DPE is $1.24$s. 
By contrast, only less $90\%$ of DAGs' completion time achieved by HEFT and FixDoc can make it. The results verify the optimality of DPE.

Fig. \ref{fig_exp2} and Fig. \ref{fig_exp3} verify the superiority of proactive stream mapping and data splitting. By spreading data 
streams over multiple links, transmission time is greatly reduced. Besides, the optimal substructure makes sure DPE can find the 
optimal placement of each function simultaneously. 

\begin{figure}[htbp]
    \centerline{\includegraphics[width=2.6in]{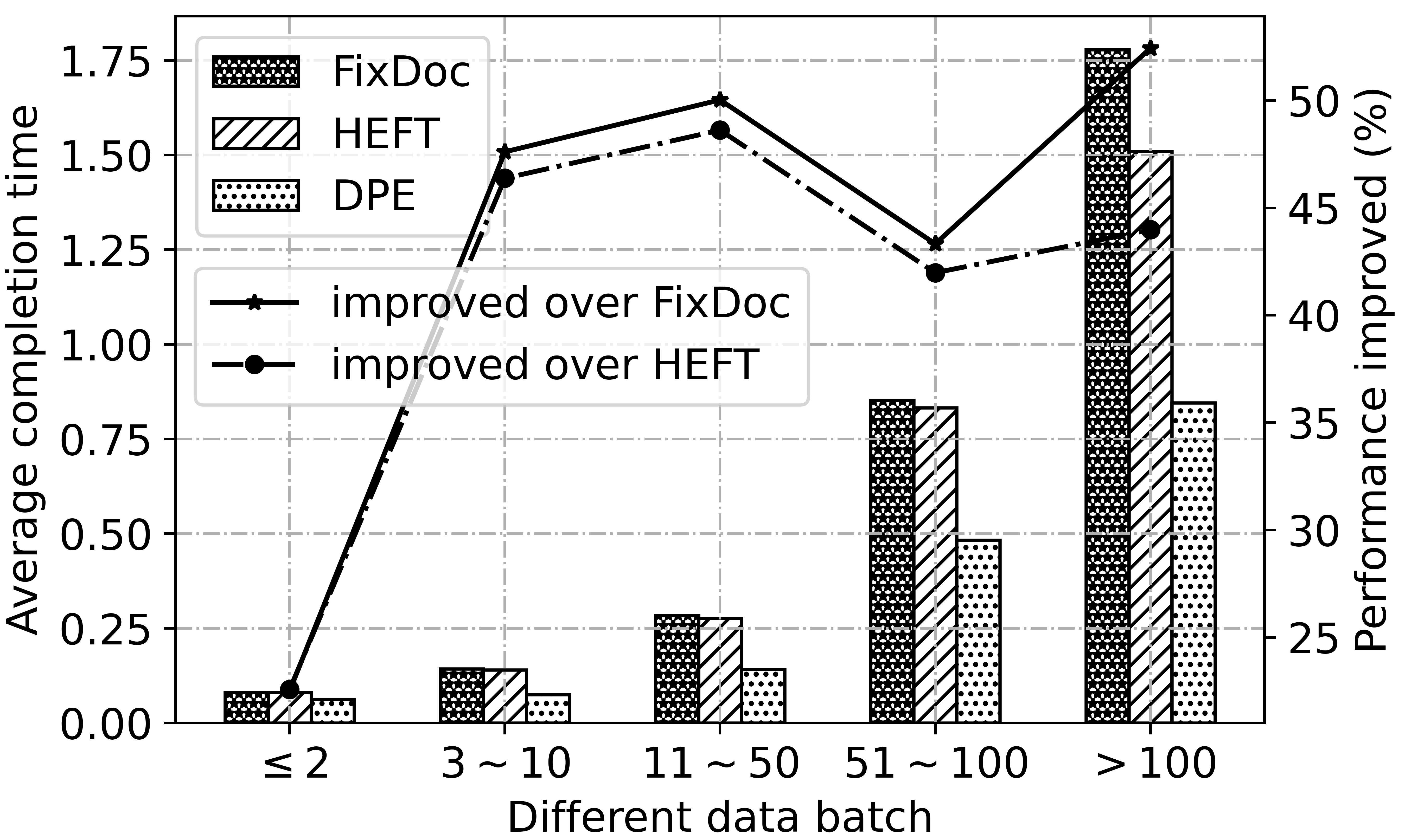}}
    \caption{Average completion time achieved by different algorithms.}
    \label{fig_exp2}
\end{figure}
\begin{figure}[htbp]
    \centerline{\includegraphics[width=2.2in]{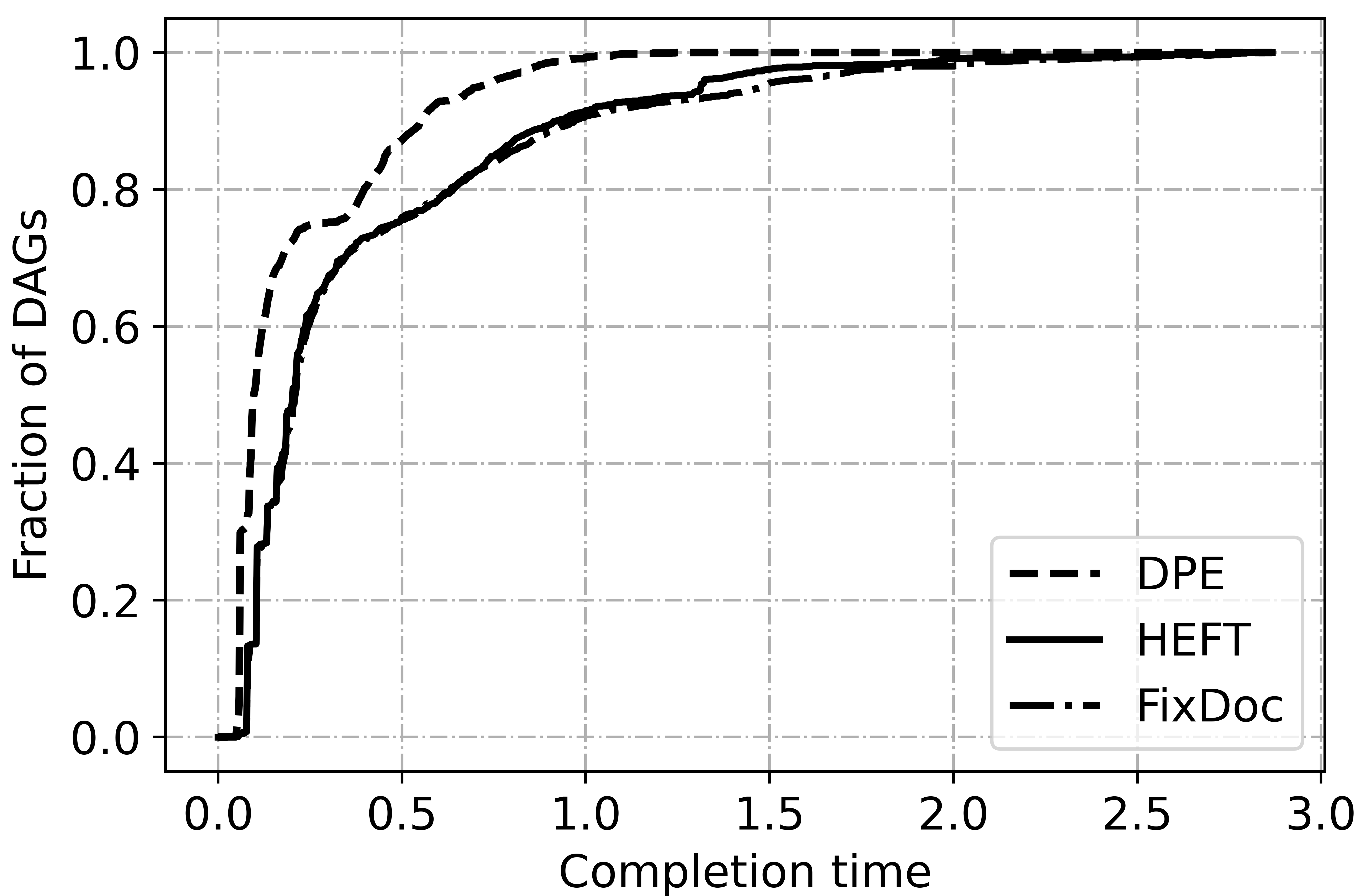}}
    \caption{CDF of completion time.}
    \label{fig_exp3}
\end{figure}

\subsubsection{Scalability Analysis}
Fig. \ref{fig_exp4} and Fig. \ref{fig_exp5} shows the impact of the scale of the heterogenous edge $\mathcal{G}$. In Fig. \ref{fig_exp4}, we 
can find that the average completion time achieved by all algorithms decreases as the edge server increases. The result is obvious 
because more \textit{idle} servers are avaliable, which ensures that more functions can be executed in parallel without waiting. For all data batches, 
DPE achieves the best result. It is interesting to find that the gap between other algorithms and DPE get widened when the scale of $\mathcal{G}$ 
increases. This is because the avaliable simple paths become more and the data transmission time is reduced even further. Fig. \ref{fig_exp4} 
also demonstrates the run time of different algorithms. The results show that DPE has the minimum time overhead. 

\begin{figure}[htbp]
    \centerline{\includegraphics[width=3.1in]{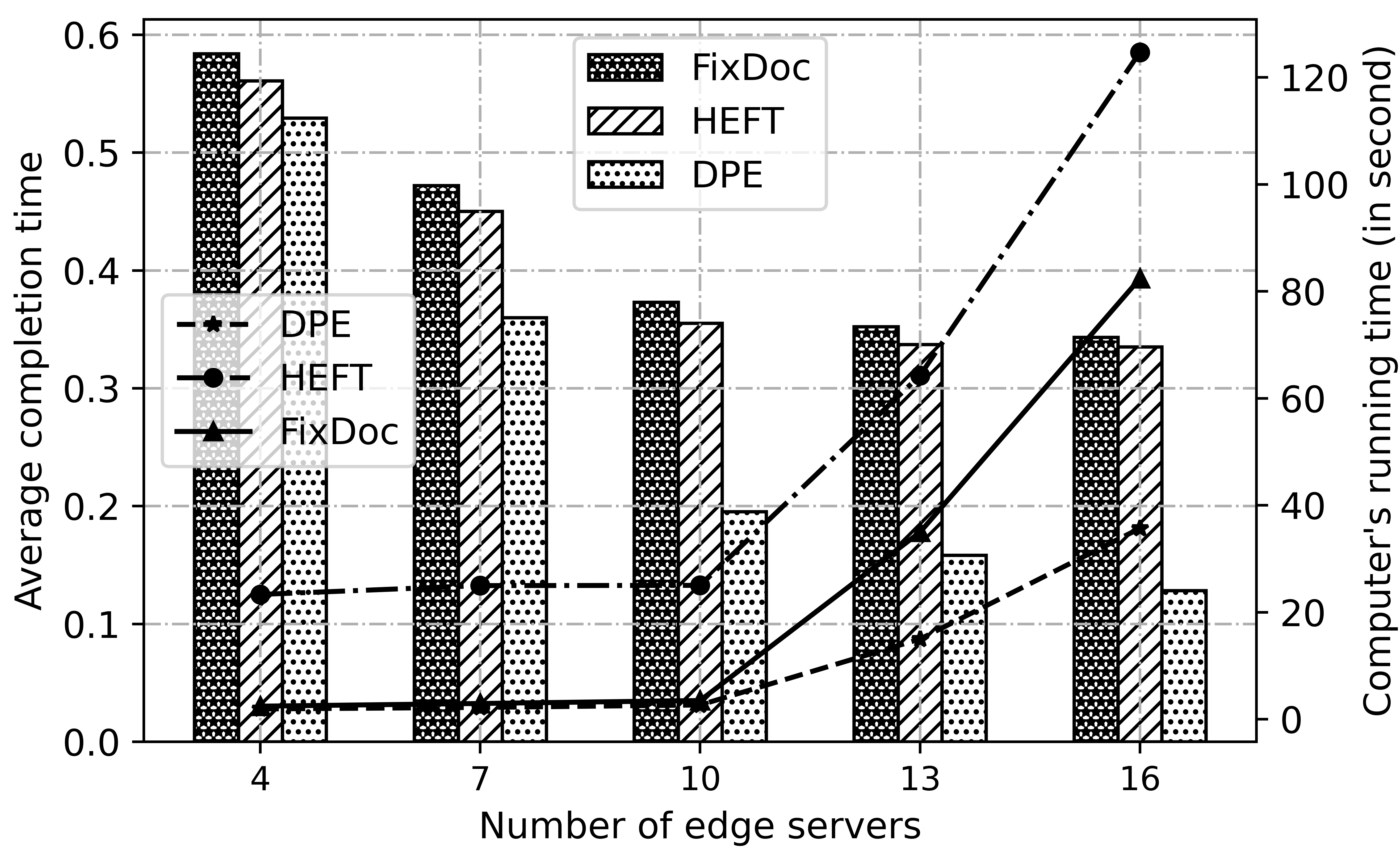}}
    \caption{Average completion time under different number of edge servers.}
    \label{fig_exp4}
\end{figure}

Fig. \ref{fig_exp5} show the impact of sparsity of $\mathcal{G}$. The horizonal axis is the overall number of simple paths $\mathcal{G}$. As 
it increases, $\mathcal{G}$ becomes more dense. because DPE can reduce transmission time with optimal data splitting and mapping, average completion 
time achieved by it decreases pretty evident. By contrast, FixDoc and HEFT have no obvious change. 

\begin{figure}[htbp]
    \centerline{\includegraphics[width=2.7in]{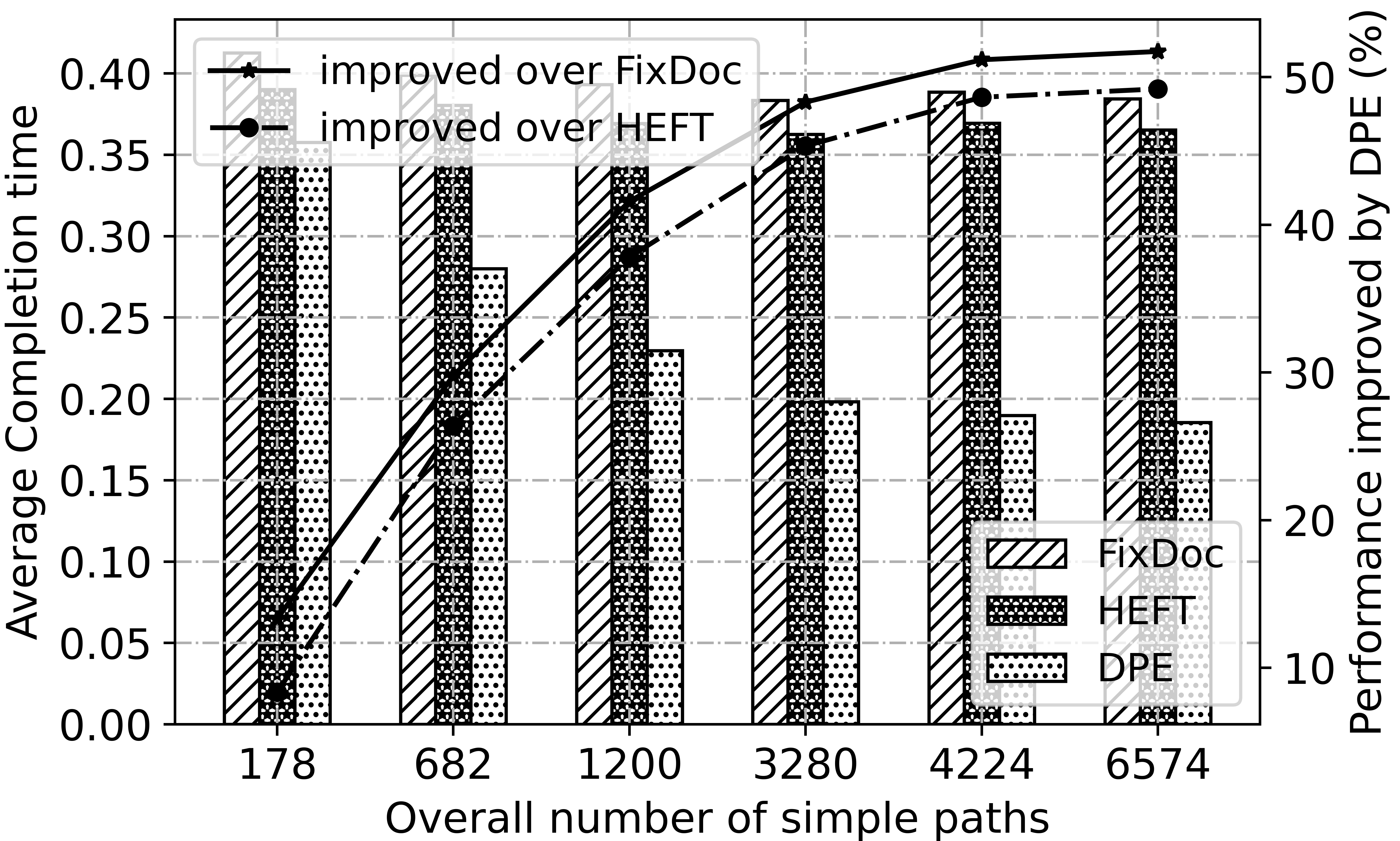}}
    \caption{Average completion time under different sparsity of $\mathcal{G}$.}
    \label{fig_exp5}
\end{figure}

\subsubsection{Sensitivity Analysis}
Fig. \ref{fig_exp6} and Fig. \ref{fig_exp7} demonstrate the impact of system parameters, $\psi_n$ and $b_l$. Notice that $\forall n,l$, 
$\psi_n$ and $b_l$ are sampled from the interval $[\psi_{lower}, \psi_{upper}]$ and $[b_{lower}, b_{upper}]$ uniformly, respectively. 
When the processing power and throughput increase, the computation and transmission time achieved by all algorithms are reduced. 
Even so, DPE outperforms all the other algorithms, which verifies the robustness of DPE adequately. 
\begin{figure}[htbp]
    \centerline{\includegraphics[width=3.4in]{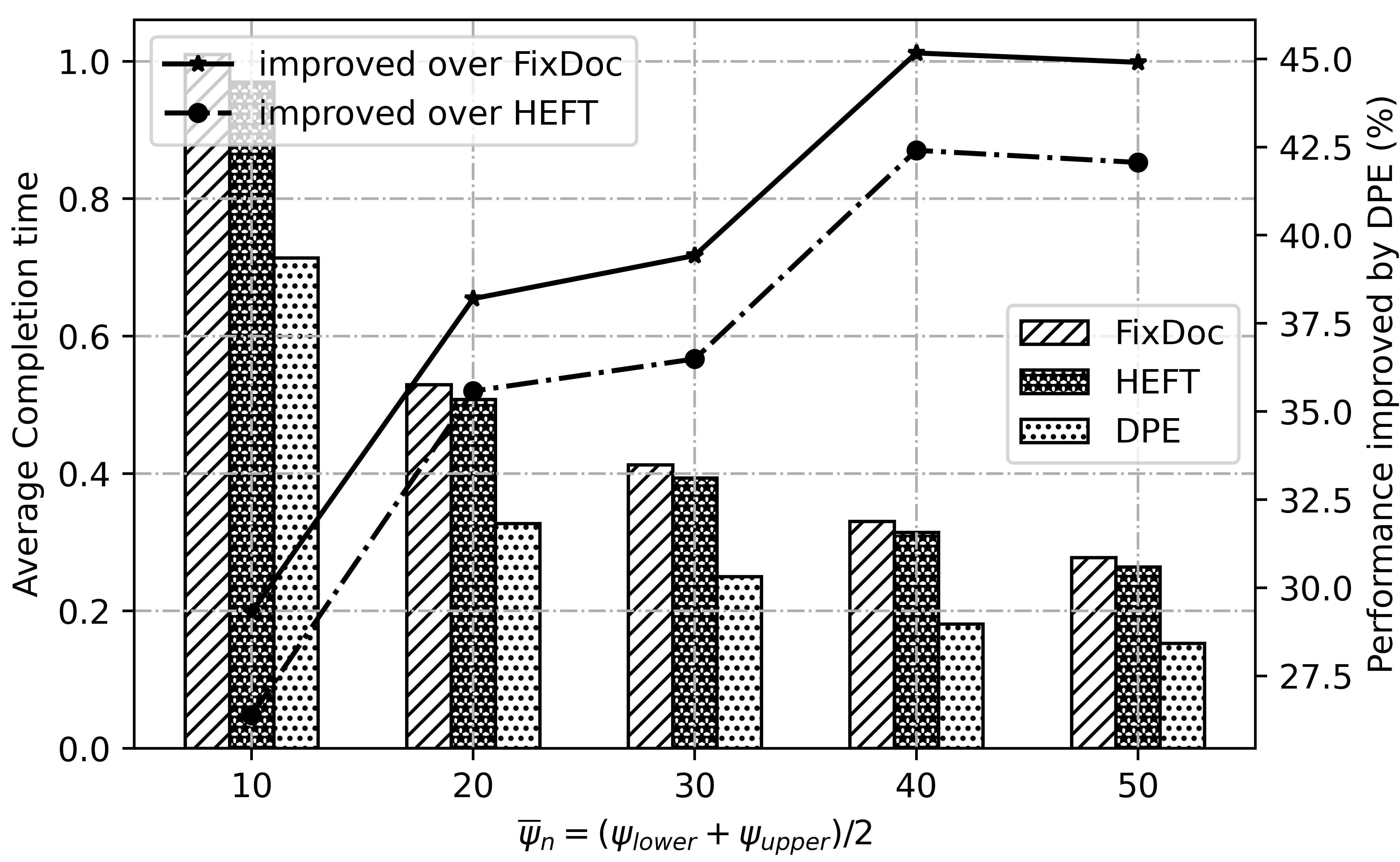}}
    \caption{Average completion time under different processing power of servers.}
    \label{fig_exp6}
\end{figure}
\begin{figure}[htbp]
    \centerline{\includegraphics[width=3.4in]{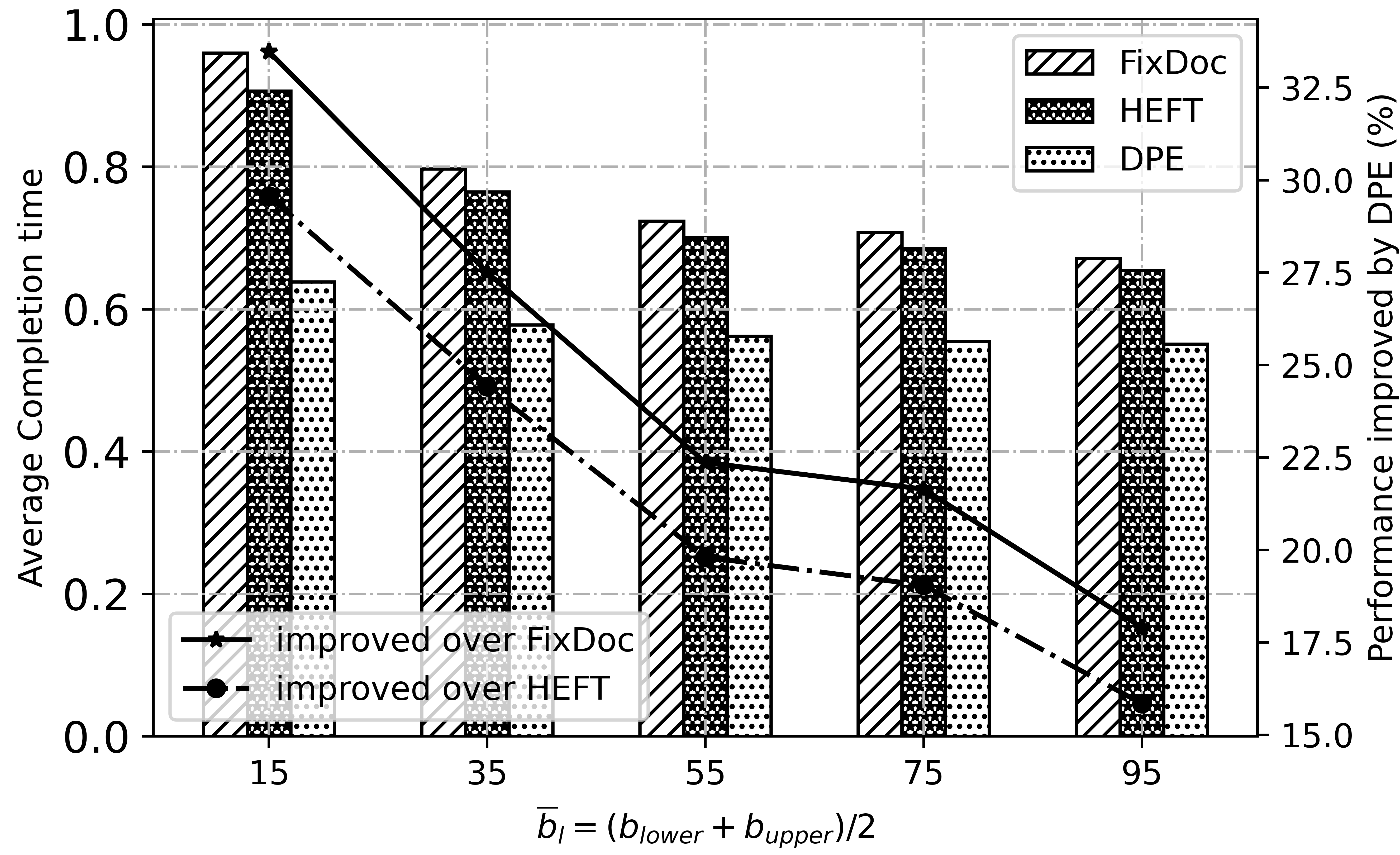}}
    \caption{Average completion time under different throughput of links.}
    \label{fig_exp7}
\end{figure}

\section{Related Works}\label{s6}
In this section, we review related works on function placement and DAG scheduling at the network edge. 

Studying the optimal function placement is not new. Since cloud computing paradigm became popular, it has been extensively studied in the 
literature \cite{scheduling-on-cloud}\cite{add2}. When bringing function placement into the paradigm of edge computing, especially for the IoT stream 
processing, different constraints, such as the response time requirement of latency-critical applications, availability of 
function instances on the heterogenous edge servers, and the wireless and wired network throughput, etc., should be taken into consideration. 
In edge computing, the optimal function placement strategy can be used to maximize the network utility \cite{ns}, minimize the inter-node 
traffic \cite{t-storm}, minimize the makespan of the applications \cite{placement1,placement2,placement3}, or even minimize the 
budget of application service providers \cite{asp}. 

The application is either modeled as an individual blackbox or a DAG with complicated composite patterns. Considering that 
the IoT stream processing applications at the edge usually have interdependent correlations between the fore-and-aft functions, 
dependent function placement problem has a strong correlation with DAG dispatching and scheduling. Scheduling algorithms for edgy computation 
tasks have been extensively studied in recent years \cite{HEFT}\cite{DAG-scheduling2,DAG-scheduling1,add1}. In edge computing, the joint 
optimization of DAG scheduling and function placement is usually NP-hard. As a result, many works can only achieve a near optimal solution 
based on heuristic or greedy policy. For example, Gedeon et al. proposed a heuristic-based solution for function placement across a three-tier 
edge-fog-cloud heterogenous infrastructure \cite{near-opt1}. Cat et al. proposed a greedy algorithm for function placement by estimating the 
response time of paths in a DAG with queue theory \cite{near-opt2}. Although FixDoc \cite{placement1} can achieve the global optimal function 
placement, the completion time can be reduced further by optimizing the stream mapping. 

\section{Conclusion}\label{s7}
This paper studies the optimal dependent function embedding problem. We first point out that proactive stream mapping and data 
splitting could have a strong impact on the makespan of DAGs with several use cases. Based on these observations, we design the 
DPE algorithm, which is theoretically verified to achieve the global optimality for an arbitrary DAG when the topological order 
of functions is given. DPE calls the RPF and the OSM algorithm to obtain the candidate paths and optimal stream mapping, respectively. 
Extensive simulations based on the Alibaba cluster trace dataset verify that our algorithms can reduce the makespan significantly 
compared with state-of-the-art function placement and scheduling methods, i.e., HEFT and FixDoc. The makespan can be further decreased by finding the 
optimal topological ordering and scheduling multiple DAGs simultaneously. We leave these extensions to our future work. 

\bibliographystyle{IEEEtran}
\bibliography{IEEEabrv,ref.bib}

\begin{thebibliography}{10}
\providecommand{\url}[1]{#1}
\csname url@samestyle\endcsname
\providecommand{\newblock}{\relax}
\providecommand{\bibinfo}[2]{#2}
\providecommand{\BIBentrySTDinterwordspacing}{\spaceskip=0pt\relax}
\providecommand{\BIBentryALTinterwordstretchfactor}{4}
\providecommand{\BIBentryALTinterwordspacing}{\spaceskip=\fontdimen2\font plus
\BIBentryALTinterwordstretchfactor\fontdimen3\font minus
  \fontdimen4\font\relax}
\providecommand{\BIBforeignlanguage}[2]{{%
\expandafter\ifx\csname l@#1\endcsname\relax
\typeout{** WARNING: IEEEtran.bst: No hyphenation pattern has been}%
\typeout{** loaded for the language `#1'. Using the pattern for}%
\typeout{** the default language instead.}%
\else
\language=\csname l@#1\endcsname
\fi
#2}}
\providecommand{\BIBdecl}{\relax}
\BIBdecl

\bibitem{placement1}
L.~Liu, H.~Tan, S.~H.-C. Jiang, Z.~Han, X.-Y. Li, and H.~Huang, ``Dependent
  task placement and scheduling with function configuration in edge
  computing,'' in \emph{Proceedings of the International Symposium on Quality
  of Service}, ser. IWQoS ’19, New York, NY, USA, 2019.

\bibitem{placement2}
S.~Khare, H.~Sun, J.~Gascon-Samson, K.~Zhang, A.~Gokhale, Y.~Barve,
  A.~Bhattacharjee, and X.~Koutsoukos, ``Linearize, predict and place:
  Minimizing the makespan for edge-based stream processing of directed acyclic
  graphs,'' in \emph{Proceedings of the 4th ACM/IEEE Symposium on Edge
  Computing}, ser. SEC ’19, New York, NY, USA, 2019, p. 1–14.

\bibitem{placement3}
Z.~{Zhou}, Q.~{Wu}, and X.~{Chen}, ``Online orchestration of cross-edge service
  function chaining for cost-efficient edge computing,'' \emph{IEEE Journal on
  Selected Areas in Communications}, vol.~37, no.~8, pp. 1866--1880, 2019.

\bibitem{SR}
F.~Duchene, D.~Lebrun, and O.~Bonaventure, ``Srv6pipes: Enabling in-network
  bytestream functions,'' \emph{Computer Communications}, vol. 145, pp. 223 --
  233, 2019.

\bibitem{slicing-algorithm-analysis}
S.~{Vassilaras}, L.~{Gkatzikis}, N.~{Liakopoulos}, I.~N. {Stiakogiannakis},
  M.~{Qi}, L.~{Shi}, L.~{Liu}, M.~{Debbah}, and G.~S. {Paschos}, ``The
  algorithmic aspects of network slicing,'' \emph{IEEE Communications
  Magazine}, vol.~55, no.~8, pp. 112--119, Aug 2017.

\bibitem{VNE-survey}
I.~{Afolabi}, T.~{Taleb}, K.~{Samdanis}, A.~{Ksentini}, and H.~{Flinck},
  ``Network slicing and softwarization: A survey on principles, enabling
  technologies, and solutions,'' \emph{IEEE Communications Surveys Tutorials},
  vol.~20, no.~3, pp. 2429--2453, thirdquarter 2018.

\bibitem{alibaba}
``Alibaba cluster trace program,''
  \url{https://github.com/alibaba/clusterdata}.

\bibitem{HEFT}
H.~{Topcuoglu}, S.~{Hariri}, and {Min-You Wu}, ``Performance-effective and
  low-complexity task scheduling for heterogeneous computing,'' \emph{IEEE
  Transactions on Parallel and Distributed Systems}, vol.~13, no.~3, pp.
  260--274, 2002.

\bibitem{scheduling-on-cloud}
G.~T. {Lakshmanan}, Y.~{Li}, and R.~{Strom}, ``Placement strategies for
  internet-scale data stream systems,'' \emph{IEEE Internet Computing},
  vol.~12, no.~6, pp. 50--60, 2008.

\bibitem{add2}
L.~Tom and V.~R. Bindu, ``Task scheduling algorithms in cloud computing: A
  survey,'' in \emph{Inventive Computation Technologies}, S.~Smys, R.~Bestak,
  and {\'A}.~Rocha, Eds.\hskip 1em plus 0.5em minus 0.4em\relax Cham: Springer
  International Publishing, 2020, pp. 342--350.

\bibitem{ns}
M.~{Leconte}, G.~S. {Paschos}, P.~{Mertikopoulos}, and U.~C. {Kozat}, ``A
  resource allocation framework for network slicing,'' in \emph{IEEE INFOCOM
  2018 - IEEE Conference on Computer Communications}, 2018, pp. 2177--2185.

\bibitem{t-storm}
J.~{Xu}, Z.~{Chen}, J.~{Tang}, and S.~{Su}, ``T-storm: Traffic-aware online
  scheduling in storm,'' in \emph{2014 IEEE 34th International Conference on
  Distributed Computing Systems}, 2014, pp. 535--544.

\bibitem{asp}
L.~{Chen}, J.~{Xu}, S.~{Ren}, and P.~{Zhou}, ``Spatio–temporal edge service
  placement: A bandit learning approach,'' \emph{IEEE Transactions on Wireless
  Communications}, vol.~17, no.~12, pp. 8388--8401, 2018.

\bibitem{DAG-scheduling2}
Y.~{Kao}, B.~{Krishnamachari}, M.~{Ra}, and F.~{Bai}, ``Hermes: Latency optimal
  task assignment for resource-constrained mobile computing,'' \emph{IEEE
  Transactions on Mobile Computing}, vol.~16, no.~11, pp. 3056--3069, 2017.

\bibitem{DAG-scheduling1}
S.~{Sundar} and B.~{Liang}, ``Offloading dependent tasks with communication
  delay and deadline constraint,'' in \emph{IEEE INFOCOM 2018 - IEEE Conference
  on Computer Communications}, 2018, pp. 37--45.

\bibitem{add1}
J.~{Meng}, H.~{Tan}, C.~{Xu}, W.~{Cao}, L.~{Liu}, and B.~{Li}, ``Dedas: Online
  task dispatching and scheduling with bandwidth constraint in edge
  computing,'' in \emph{IEEE INFOCOM 2019 - IEEE Conference on Computer
  Communications}, 2019, pp. 2287--2295.

\bibitem{near-opt1}
J.~{Gedeon}, M.~{Stein}, L.~{Wang}, and M.~{Muehlhaeuser}, ``On scalable
  in-network operator placement for edge computing,'' in \emph{2018 27th
  International Conference on Computer Communication and Networks (ICCCN)},
  2018, pp. 1--9.

\bibitem{near-opt2}
X.~Cai, H.~Kuang, H.~Hu, W.~Song, and J.~L{\"u}, ``Response time aware operator
  placement for complex event processing in edge computing,'' in
  \emph{Service-Oriented Computing}, C.~Pahl, M.~Vukovic, J.~Yin, and Q.~Yu,
  Eds.\hskip 1em plus 0.5em minus 0.4em\relax Cham: Springer International
  Publishing, 2018, pp. 264--278.

\end{thebibliography}

\end{document}